\documentclass[letterpaper,11pt]{article}
\usepackage[utf8]{inputenc}

\usepackage{amsmath, amsthm, amssymb, thm-restate}
\usepackage{dsfont}
\usepackage{algorithmicx}
\usepackage[table,xcdraw]{xcolor}
\usepackage[ruled,vlined,linesnumbered]{algorithm2e}
\usepackage{setspace}
\usepackage{mathtools}
\usepackage{pgf}
\usepackage{ifthen} 

\usepackage[style=alphabetic,natbib=true,maxnames=99,maxalphanames=5]{biblatex}
\usepackage{comment} 
\usepackage[most]{tcolorbox}   
\usepackage{xfrac}
\usepackage{hyperref}
\usepackage{multirow}
\usepackage{caption}
\usepackage{bm}
\usepackage{newfloat}
\usepackage{enumitem}
\usepackage{wrapfig}
\usepackage{dsfont}
\usepackage{multirow}
\usepackage{etoolbox}          
\usepackage{fancyhdr}
\usepackage{csquotes}
\usepackage[noabbrev]{cleveref}
\usepackage{tabularx}
\usepackage[table]{xcolor} 

\usepackage[margin=1in]{geometry}

\allowdisplaybreaks

\definecolor{mygreen}{RGB}{10,110,230}
\definecolor{myred}{RGB}{10,110,230}

\hypersetup{
     colorlinks=true,
     citecolor=mygreen,
     linkcolor=myred
}

\renewcommand{\epsilon}{\varepsilon}

\newcommand{\sfcomment}[1]{}
\newcommand{\sbcomment}[1]{}
\newcommand{\aacomment}[1]{}

\newcommand{\hiddencomment}[1]{}

\newcommand{\mc}[1]{\ensuremath{\mathcal{#1}}}

\newcounter{protocolcounter}
 
\crefname{protocolcounter}{Algorithm}{Algorithms}

\newcommand{\pos}{\operatorname{pos}}

\newcommand{\val}[1]{\mathsf{val}_{#1}}
\newcommand{\vallp}[1]{\mathsf{val}^{\mathsf{LP}}_{#1}}
\newcommand{\aloc}{\mathcal{A}_{\mathsf{loc}}}

\newcommand{\B}{\textbf{B}}
\newcommand{\C}{\textbf{C}}

\newcommand{\Ex}{\mathbb{E}}
\newcommand{\Tmax}{T_{max}}

\newcommand{\Var}{\operatorname{Var}}
\newcommand{\Cov}{\operatorname{Cov}}
\newcommand{\Out}{\mathsf{Out}}

\newcommand{\card}[1]{|#1|}

\newcommand{\degs}{\mathsf{degs}}

\newcommand{\dest}{\widetilde{d}}
\newcommand{\total}{\mathsf{total}}
\newcommand{\StreamingReduction}{\mathsf{StreamingReduction}}
\newcommand{\Aggregate}{\mathsf{Aggregate}}
\newcommand{\InitReservoir}{\mathsf{InitializeReservoir}}
\newcommand{\UpdateReservoir}{\mathsf{UpdateReservoir}}

\newcommand{\on}{\mathsf{on}}
\newcommand{\off}{\mathsf{off}}
\newcommand{\eadv}{\epsilon_{\mathsf{adv}}}
\newcommand{\bounded}{\mathsf{bdd}}

\newcommand{\maxdicut}[0]{\text{\textsc{Max-DiCut}}}
\newcommand{\CSP}[0]{\textsc{CSP}}
\newcommand{\maxCSP}[0]{\textsc{Max-CSP}}
\newcommand{\BasicLP}{\textsc{BasicLP}}
\newcommand{\LOCAL}{\ensuremath{\mathsf{LOCAL}}}

\newcommand{\StreamingEstimator}{\mathsf{StreamingEstimator}}
\newcommand{\Gest}{\widetilde{G}}

\newcommand{\valest}{\ensuremath{\widetilde{\val{I}}}}
\newcommand{\one}{\mathds{1}}
\SetKwRepeat{Do}{do}{while}

\DeclarePairedDelimiter{\paren}{\lparen}{\rparen}
\DeclarePairedDelimiter{\set}{\lbrace}{\rbrace}

\DeclarePairedDelimiter{\bracket}{[}{]}
\DeclarePairedDelimiter{\abs}{\lvert}{\rvert}

\SetKwProg{WithProb}{with probability}{:}{}

\DeclareMathOperator{\poly}{poly}

\renewcommand{\epsilon}[0]{\ensuremath{\varepsilon}}

\let\originalleft\left
\let\originalright\right
\renewcommand{\left}{\mathopen{}\mathclose\bgroup\originalleft}
\renewcommand{\right}{\aftergroup\egroup\originalright}

\newtheorem{theorem}{Theorem}

\newtheorem{lemma}{Lemma}[section]

\newtheorem{proposition}[lemma]{Proposition}

\newtheorem{conj}{Conjecture}
\newtheorem{definition}[lemma]{Definition}
\newtheorem{claim}[lemma]{Claim}

\newtheorem{observation}[lemma]{Observation}

\makeatletter
\def\thm@space@setup{%
  \thm@preskip= 0.2cm
  \thm@postskip=\thm@preskip 
}
\makeatother

\crefname{lemma}{Lemma}{Lemmas}
\crefname{theorem}{Theorem}{Theorems}
\crefname{property}{Property}{Properties}
\crefname{claim}{Claim}{Claims}
\crefname{corollary}{Corollary}{Corollaries}
\crefname{result}{Result}{Results}
\crefname{conj}{Conjecture}{Conjectures}
\crefname{definition}{Definition}{Definitions}
\crefname{observation}{Observation}{Observations}
\crefname{proposition}{Proposition}{Propositions}
\crefname{assumption}{Assumption}{Assumptions}
\crefname{line}{Line}{Lines}
\crefname{figure}{Figure}{Figures}
\creflabelformat{property}{(#1)#2#3}
\crefname{equation}{}{}
\crefname{section}{Section}{Sections}
\crefname{appendix}{Appendix}{Appendices}
\crefname{problem}{Problem}{Problems}
\crefname{algorithm}{Algorithm}{Algorithms}

\definecolor{mygreen}{RGB}{20,155,20}
\definecolor{myred}{RGB}{195,20,20}
\definecolor{linkcolor}{RGB}{0,0,230}
\definecolor{mylightgray}{RGB}{230,230,230}
\definecolor{verylightgray}{RGB}{240,240,240}
\definecolor{commentcolor}{RGB}{120,120,120}

\SetCommentSty{mycommfont}

\hypersetup{
     colorlinks=true,
     citecolor= mygreen,
     linkcolor= myred,
     urlcolor= mygreen
}

\renewcommand{\mc}[1]{\ensuremath{\mathcal{#1}}}

\newcounter{myalgctr}

\newenvironment{tbox}{
\par\addvspace{0.2cm}
\begin{tcolorbox}[width=\textwidth,
                  boxsep=2pt,
                  left=1pt,
                  right=1pt,
                  top=4pt,
                  boxrule=1pt,
                  arc=0pt,
                  colback=white,
                  colframe=black
                  ]
}{
\end{tcolorbox}
}

\newenvironment{tboxh}{
\par\addvspace{0.2cm}
\begin{tcolorbox}[width=\textwidth,
                  boxsep=2pt,
                  left=1pt,
                  right=1pt,
                  top=4pt,
                  boxrule=1pt,
                  arc=0pt,
                  colback=white,
                  colframe=black,
                  float=t
                  ]
}{
\end{tcolorbox}
}

\newenvironment{graytbox}{
\par\addvspace{0.1cm}
\begin{tcolorbox}[width=\textwidth,
                  frame hidden,
                  boxsep=5pt,
                  left=1pt,
                  right=1pt,
                  top=2pt,
                  bottom=2pt,
                  boxrule=1pt,
                  arc=0pt,
                  colback=mylightgray,
                  colframe=black,
                  breakable
                  ]
}{
\end{tcolorbox}
}

\newcommand{\tboxhrule}[0]{\vspace{0.1cm} \hrule \vspace{0.2cm}}

\newenvironment{titledtbox}[1]{\begin{tbox}#1 \tboxhrule}{\end{tbox}}
\newenvironment{titledtboxh}[1]{\begin{tboxh}#1 \tboxhrule}{\end{tboxh}}

\definecolor{mylightblue}{RGB}{225,242,250}
\definecolor{myblueborder}{RGB}{120,180,210}

\newtcolorbox{lightbluebox}[1][]{
  colback=mylightblue,
  colframe=myblueborder,
  coltitle=black,
  boxrule=1.2pt,
  arc=0mm,
  left=6pt,
  right=6pt,
  top=6pt,
  bottom=6pt,
  fonttitle=\bfseries,
  title=#1
}

\makeatletter
\renewcommand{\paragraph}{%
  \@startsection{paragraph}{4}%
  {\z@}{10pt}{-1em}%
  {\normalfont\normalsize\bfseries}%
}
\makeatother

\makeatletter
\patchcmd{\@algocf@start}{-1.5em}{0pt}{}{}
\makeatother

\title{Single-Pass Streaming CSPs via Two-Tier Sampling}

\author{
Amir Azarmehr\thanks{Northeastern University. Emails: \texttt{\{azarmehr.a, s.behnezhad, ferrante.s\}@northeastern.edu}. This work was supported in part by NSF CAREER Award CCF-2442812, Soheil Behnezhad's Alfred P. Sloan Fellowship, and a Google
Research Award.} \and
Soheil Behnezhad\footnotemark[1]  \and
Shane Ferrante\footnotemark[1]
}

\begin{document}

\date{}

\maketitle

\thispagestyle{empty}

\begin{abstract}
{
\setlength{\parskip}{0.5em}
We study the maximum constraint satisfaction problem, \textsc{Max-CSP}, in the streaming setting. Given $n$ variables, the constraints arrive sequentially in an arbitrary order, with each constraint involving only a small subset of the variables. The objective is to approximate the maximum fraction of constraints that can be satisfied by an optimal assignment in a single pass. The problem admits a trivial near-optimal solution with $O(n)$ space, so the major open problem in the literature has been the best approximation achievable when limiting the space to $o(n)$.

The answer to the question above depends heavily on the \textsc{CSP} instance at hand. The integrality gap $\alpha$ of an LP relaxation, known as the \textsc{BasicLP}, plays a central role. In particular, a major conjecture of the area is that in the single-pass streaming setting, for any fixed $\varepsilon > 0$,
\begin{enumerate}[label=($\roman*$), topsep=5pt]
    \item an $(\alpha-\varepsilon)$-approximation can be achieved with $o(n)$ space, and that
    \item any $(\alpha+\varepsilon)$-approximation requires $\Omega(n)$ space.
\end{enumerate}

Despite significant progress, both sides of the conjecture remain open. Under the constant degree assumption \cite[arXiv'25]{SingerTV25} or by allowing multiple passes \cite[STOC'26]{FMW25dichotomy}, the first part has been positively resolved. The latter work also gives the best known space lower bound of $\Omega(n^{1/3})$ for the second part of the conjecture. We also note that for the special case of \textsc{Max-DiCut}, where the integrality gap is $\alpha = 1/2$, both sides of the conjecture have been resolved due to the works of \cite[STOC'26]{AzarmehrBFS26} and \cite[STOC'19]{KapralovK19}.

In this work, we fully resolve the first side of the conjecture by proving that an $(\alpha - \varepsilon)$-approximation of \textsc{Max-CSP} can indeed be achieved using $n^{1-\Omega_\varepsilon(1)}$ space and in a single pass.  Given that \textsc{Max-DiCut} is a special case of \textsc{Max-CSP}, our algorithm fully recovers the recent result of \cite[STOC'26]{AzarmehrBFS26} via a completely different algorithm and proof. On a technical level, our algorithm simulates a suitable local algorithm on a reduced graph using a technique that we call {\em two-tier sampling}: the algorithm combines both edge sampling and vertex sampling to handle high- and low-degree vertices at the same time.
}
\end{abstract}

{
\newpage
\tableofcontents
\thispagestyle{empty}
}

\clearpage
\setcounter{page}{1}

\section{Introduction}

\emph{Constraint satisfaction problems} (\CSP s) concern assigning values from a finite domain to $n$ variables so as to maximize the number of satisfied constraints, where each constraint is specified on a small subset of the variables.
They serve as a foundational abstraction for a vast class of combinatorial problems, unifying many prominent optimization problems, including maximum cut for directed and undirected graphs, maximum $q$-colorability, the Unique Games problem, etc.
Their study has been instrumental in major developments of algorithmic and complexity-theoretic tools, including semidefinite approximation algorithms, the PCP theorem, the Unique Games conjecture, and modern analytic methods for Boolean functions. In this paper, we study space-efficient algorithms for CSPs particularly in the streaming setting.

Formally, for an integer $k > 0$ and an alphabet $\Sigma$, \emph{a maximum constraint satisfaction problem}, $\CSP(\mc{F})$, is defined by a set of predicates $\mc{F} \subseteq \Sigma^k \to \{0, 1\}$.
An instance $I = (V, F)$ is given by $m$ constraints $F$ from the set of predicates $\mc{F}$ over $n$ variables $V$. The value of an assignment $\sigma \in \Sigma^n$ is equal to the fraction of constraints it satisfies, and is denoted by $\val{I}(\sigma) \in [0, 1]$.
The maximum constraint satisfaction problem, $\maxCSP(\mc{F})$, asks to compute or approximate the value of the optimal assignment $\val{I} := \max_{\sigma} \val{I}(\sigma)$.

We study approximating \maxCSP{} in the streaming setting, which has received extensive attention over the past decade \cite{KapralovKS15,GuruswamiVV17,KapralovK19,ChouGV20,AssadiKSY20,AssadiN21,ChouGSVV22,SaxenaSSV23FOCS,HwangSV24,ChouGSV24,SaxenaSSV25,FeiMW25FOCS,FMW25dichotomy,velusamy2025near,SingerTV25,AzarmehrBFS26}. We particularly refer interested readers to the excellent survey of \citet*{Sudan22} on this problem and its importance.
In the streaming setting, the constraints arrive in an arbitrary (e.g., adversarial chosen) order and must be processed by an algorithm with limited space that ideally takes one pass over this stream.
The interesting space regime is $o(n)$, since randomly storing $O(n/\epsilon^2)$ constraints easily yields a $(1 - \epsilon)$-approximation, rendering linear-memory approximations trivial.\footnote{This result is folklore; see e.g. \cref{claim:linear-edges} in the present paper.}

A standard approach to approximating \maxCSP{} is via LP relaxations.
For any specific $\CSP(\mc{F})$, a natural relaxation, canonically referred to as the \BasicLP{}, has been shown to be essentially optimal amongst polynomially-sized LPs \cite{GhoshT17,ChanLRS16,KothariMR22}.
The integrality gap of the \BasicLP{}, hereafter denoted by $\alpha_{\mc{F}}$, is closely linked with the efficient approximability of $\maxCSP{(\mc{F})}$ across multiple models. In particular, a major conjecture of the area is the following dichotomy for single-pass streaming algorithms:

\begin{lightbluebox}
\begin{conj}[\textbf{Streaming Dichotomy}]\label{conj}
For any given $\CSP(\mc{F})$ and every fixed $\epsilon > 0$, 
\begin{itemize}[itemsep=5pt, topsep=10pt, leftmargin=10pt]
    \item an $(\alpha_{\mc{F}} - \epsilon)$-approximation of $\maxCSP(\mc{F})$ can be obtained with $o(n)$ space in a single pass,
    \item and an $(\alpha_{\mc{F}} + \epsilon)$-approximation of $\maxCSP(\mc{F})$ in a single pass requires $\Omega(n)$ space.
\end{itemize}
\end{conj}
\end{lightbluebox}

\Cref{conj} has been mentioned implicitly and explicitly in several works in the literature \cite{Sudan22,FMW25dichotomy,SingerTV25}. See for example \cite{SingerTV25} which is particularly about this conjecture.

While both sides of \Cref{conj} remain open, some major progress has been made towards proving it. For example, under the constant degree assumption (i.e., each variable is involved in a constant number of constraints) \citet{SingerTV25} fully resolved the first part. Additionally, if instead of a single pass, one allows multiple passes the recent algorithm of \citet{FMW25dichotomy} substantially resolves the first, leaving the single-pass case open:
\begin{quote}
\emph{\enquote{[...] it remains largely mysterious whether CSP approximation exhibits a dichotomy behavior in the single-pass model as well.}} \cite{FMW25dichotomy}
\end{quote}

On the lower bound side, the work of \cite{FMW25dichotomy} shows $\Omega(n^{1/3}/p)$ space is needed for an $(\alpha_{\mc{F}}+ \epsilon)$-approximation in $p$ passes. While not quite linear in $n$, this already implies a dichotomy for multi-pass algorithms since their algorithm uses only polylogarithmic space for an $(\alpha_{\mc{F}}-\epsilon)$-approximation.

It is also worth noting that for the spacial case of \textsc{Max-DiCut}, for which the integrality gap is $\alpha_{\mc{F}} = 1/2$, both sides of the conjecture have been fully resolved due to the recent upper bound of \citet{AzarmehrBFS26} that takes $n^{1-\Omega_\epsilon(1)}$ space in a single-pass, and an earlier lower bound of $\Omega(n)$ due to \citet{KapralovK19}.

\paragraph{Our Contribution:} Our main contribution in this work is to completely resolve the first part of \Cref{conj} by proving the following theorem:
\begin{graytbox}
\begin{theorem}\label{thm:main}
For an instance $\CSP(\mc{F})$, let $\alpha_{\mc{F}}$ denote the integrality gap of the corresponding \linebreak \BasicLP{}. Then for any $\epsilon > 0$, there is a randomized single-pass streaming algorithm that w.h.p. $(\alpha_{\mc{F}} - \epsilon)$-approximates $\maxCSP{(\mc{F})}$ using $n^{1 - \Omega_\epsilon(1)}$ space.
\end{theorem}
\end{graytbox}

Given that \textsc{Max-DiCut} is a special case of \textsc{Max-CSP}, our algorithm fully recovers the recent result of \cite{AzarmehrBFS26} via a completely different algorithm and proof. See \Cref{sec:techniques} for an overview of our techniques and a comparison to previous ones.
\section{Technical Overview}\label{sec:techniques}

\subsection{Background}

Our starting point is the work of \cite{SingerTV25}, which solves constant-degree CSPs in the streaming setting. This algorithm combines the distributed algorithm of \cite{Yoshida11} 
with the streaming techniques of \cite{SaxenaSSV25} for approximating the \maxdicut{} in constant-degree graphs,
to show that the basic LP for constant-degree \CSP s can be approximated in sublinear space.
Here, the degree of a variable in a \CSP{} denotes the number of constraints it appears in, and the \CSP{} is bounded-degree if all its variables have a bounded degree. 
A well-known reduction by \citet*{trevisan} converts any \CSP{} to one with the same value that has a bounded maximum degree. At first glance, this might suggest that we are done: apply the reduction of \cite{trevisan} to reduce the degrees, and then employ the constant-degree algorithm of \cite{SingerTV25}. However, the main challenge is that this reduction requires linear space in the number of constraints, far beyond the amount of space available in the streaming setting. In fact, this is the main challenge that recent works of \cite{AzarmehrBFS26,velusamy2025near} try to address for the \maxdicut{} problem, which is a special case of CSPs. 

Our algorithm runs the LP on a modified variant of Trevisan's reduced graph where the degrees are bounded only in expectation.
However, this variant satisfies additional statistical properties that allow us to essentially ignore any constraints adjacent to variables violating the constant-degree bound. An offline (non-streaming) description is presented below.

Given a \CSP{} with $m$ constraints and $n$ variables, the reduction proceeds by creating a new instance with $B$ constraints corresponding to each of the original constraints (hereafter referred to as the corresponding copies of the original constraint), and $\deg(v)$ variables corresponding to each original variable.
Here, $B$ is a sufficiently large constant determined by our accuracy parameter $\epsilon$.
For each constraint $C$ in the original \CSP{}, the corresponding copies use the same predicate as $C$.
The variables involved in each copy of $C$ shall be copies of the variables involved in $C$.
Let $v_i$ denote $i$-th variable appearing in $C$.
Then, each copy of $C$ chooses a copy of $v_i$ uniformly at random for its $i$-th variable. See the figure below for an illustration with $B=2$, where the vertices correspond to variables and the edges correspond to constraints.

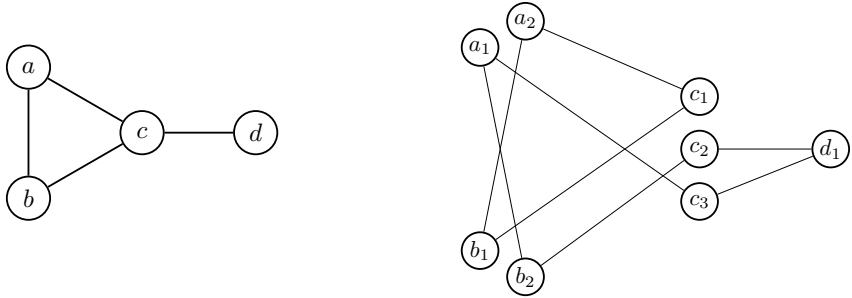
\begin{figure}[h]
    \centering
    \resizebox{0.7\textwidth}{!}{
    \begin{tikzpicture}[  
    every node/.style={circle, draw, thick, inner sep=0pt},
    orig/.style={minimum size=0.66cm},
    redc/.style={minimum size=0.56cm, font=\small},
    thick
]

\def\graphsep{8.0}      
\def\alpha{30}          
\def\rspread{0.8}       
\def\rvgap{3.5}         
\def\rhgap{3.0}         
\def\dhgap{2.0}         
\def\rcdoffset{0}       
\def\B{2}               
\pgfmathsetseed{456}    

\gdef\drawnedges{}
\foreach \char in {a,b,c,d} {
    \foreach \num in {1,2,3} {
        \expandafter\gdef\csname deg\char\num\endcsname{0}
    }
}

\pgfmathsetmacro{\rdx}{\rspread * cos(\alpha)}
\pgfmathsetmacro{\rdy}{\rspread * sin(\alpha)}
\pgfmathsetmacro{\halfsep}{\graphsep / 2}

\def\origwidth{3.464} 
\pgfmathsetmacro{\origshift}{-\halfsep - (\origwidth / 2)}

\pgfmathsetmacro{\redwidth}{\rhgap + \dhgap + \rcdoffset}
\pgfmathsetmacro{\redshift}{\halfsep - (\redwidth / 2)}

\begin{scope}[xshift=\origshift cm, yshift=1.0cm, nodes={orig}]
    \def\hgap{1.732} \def\vgap{2.0}
    \node (a) at (0, \vgap)    {$a$};
    \node (b) at (0, 0)        {$b$};
    \node (c) at (\hgap, \vgap/2)  {$c$};
    \node (d) at (2*\hgap, \vgap/2)  {$d$};
    \draw (a) -- (b); \draw (a) -- (c); \draw (b) -- (c); \draw (c) -- (d);
\end{scope}

\begin{scope}[xshift=\redshift cm, nodes={redc}]
    
    \foreach \i [evaluate=\i as \pos using (\i-1.5)] in {1,2} { \node (a\i) at (\pos*\rdx, \rvgap + \pos*\rdy) {$a_{\i}$}; }
    \foreach \i [evaluate=\i as \pos using (\i-1.5)] in {1,2} { \node (b\i) at (\pos*\rdx, 0 - \pos*\rdy) {$b_{\i}$}; }
    \foreach \i [evaluate=\i as \pos using (2.0-\i)] in {1,2,3} { \node (c\i) at (\rhgap + \rcdoffset, \rvgap/2 + \pos*\rspread) {$c_{\i}$}; }
    \foreach \i [evaluate=\i as \pos using (1.0-\i)] in {1} { \node (d\i) at (\rhgap + \dhgap + \rcdoffset, \rvgap/2 + \pos*\rspread) {$d_{\i}$}; }

    \begin{scope}[draw=black, thin]
        \newcommand{\drawRandomBRegular}[4]{ 
            \def\edgecount{0}
            \whiledo{\edgecount < \B}{
                \pgfmathsetmacro{\u}{int(random(1,#2))}
                \pgfmathsetmacro{\v}{int(random(1,#4))}
                \edef\edgeid{e#1\u#3\v}
                \def\isnew{1}
                \foreach \entry in \drawnedges { \ifx\entry\edgeid \def\isnew{0} \fi }
                \expandafter\let\expandafter\curdegU\csname deg#1\u\endcsname
                \expandafter\let\expandafter\curdegV\csname deg#3\v\endcsname
                \ifnum\isnew=1
                    \ifnum\curdegU<\B
                        \ifnum\curdegV<\B
                            \draw (#1\u) -- (#3\v);
                            \xdef\drawnedges{\drawnedges \edgeid,}
                            \pgfmathsetmacro{\newdegU}{int(\curdegU + 1)}
                            \pgfmathsetmacro{\newdegV}{int(\curdegV + 1)}
                            \expandafter\xdef\csname deg#1\u\endcsname{\newdegU}
                            \expandafter\xdef\csname deg#3\v\endcsname{\newdegV}
                            \pgfmathsetmacro{\nextcount}{int(\edgecount + 1)}
                            \global\let\edgecount=\nextcount
                        \fi
                    \fi
                \fi
            }
        }
        \drawRandomBRegular{a}{2}{b}{2}
        \drawRandomBRegular{a}{2}{c}{3}
        \drawRandomBRegular{b}{2}{c}{3}
        \drawRandomBRegular{c}{3}{d}{1}
    \end{scope}
\end{scope}

\end{tikzpicture}
    }
    \label{fig:placeholder}
\end{figure}
\noindent
Note that each variable copy in the reduced instance has {\em expected} degree $B$ which is a constant.

There are a number of major challenges to implement this reduction in the streaming setting and then apply the algorithm of \cite{SingerTV25} that we address in the next section. First, the number of copies of each variable (vertex) depends on its degree, which is only revealed gradually in the stream. Moreover, the existence of ``high-degree'' vertices in the original graph makes it challenging to collect the relevant local neighborhood of the corresponding vertices in the reduced graph, which is needed to simulate the distributed algorithm of \cite{Yoshida11}.

\subsection{Our contribution}

Throughout this section, we assume that all the predicates are $k$-ary (take $k$ variables) and non-trivial (they are not satisfied or violated independently of the variable assignments), and that the number of constraints $m$ is linear in the number of variables $n$. These assumptions hold w.l.o.g.\ by simply sub-sampling the constraints with the appropriate probability; see \cref{sec:input-simplification}.

On a high level, our algorithm estimates the value of the \BasicLP{} by sampling a subset of the constraints and approximating their contribution to the value of the LP.
To do so, we attempt to collect the $r$-neighborhood around each sampled constraint and run the \LOCAL{} algorithm due to \cite{Yoshida11}, where $r$ is a constant depending on $\epsilon$.

For constant-degree \CSP s, \cite{SingerTV25,SaxenaSSV25} show that this can be done by sampling every variable with probability of $n^{-c}$ and storing all the constraints they appear in.
Since the \CSP{} has a constant maximum degree $\Delta$, the size of the $r$-neighborhood is at most a constant $\Tmax = (k\Delta)^r$.
Therefore, the entire neighborhood is sampled with a probability $n^{-c\Tmax}$, which turns out to be sufficiently large when $c$ is relatively small.

This scheme fails when the degrees are unbounded, since 
(1) the constraints adjacent to a high-degree variable may not fit into the memory, and more importantly,
(2) the sizes of the $r$-neighborhoods can be arbitrarily large, and hence there is no hope of collecting the entire neighborhood by sampling.

This is where Trevisan's reduction is employed to construct a \CSP{} with bounded degree.
However, as mentioned earlier, Trevisan's reduction has a size linear in the number of constraints, and hence cannot be stored in the memory.
Furthermore, the construction depends on the degrees of vertices, which is not known to the algorithm through the stream.
Hence, it cannot be constructed trivially in the streaming setting.

The crux of our algorithm is to access parts of the reduced \CSP{}, without ever explicitly building the rest of it.
These parts are then used to run the low-degree scheme and obtain the approximation of \maxCSP{}, all in one pass.
We achieve this through our main technical contribution, which we call \emph{two-tier sampling}.
The end goal is to sample an $n^{-c}$-fraction of the variables in the reduced \CSP, and store all the constraints they appear in.
To do so, we divide the variables of the original \CSP{} into two tiers, low-degree and high-degree, and handle them differently.
The low-degree variables are managed by sampling the variables, whereas the high-degree variables are treated by sampling the constraints.
We use $n^q$ as a threshold for the degree\footnote{In reality, our true threshold is $n^q$ on the number of samples we see which actually implies an expected degree of $n^{q+c}/B$ based on the construction. This technical detail will be elaborated upon further in the analysis.}, where $q$ is a sufficiently small constant determined by $\epsilon$. This threshold is imposed only approximately in the final streaming algorithm, as we use sampling to detect the high-degree variables. 
Any low-degree variables that are erroneously detected as high-degree shall hereafter simply be referred to as high-degree.

Each low-degree variable (in the original $\CSP$) is sampled with a probability of $n^{-c}$, and all adjacent constraints are stored for the sampled variables.
This fits in a memory of $O(n^{1 - c})$, since each constraint is stored if any of its $k$ variables are sampled, which happens with a probability of at most $k n^{-c}$, and the number of constraints is $O(n)$.
Using this information, we can access \emph{all} the copies of the sampled variables in the reduced \CSP{}:
We create $B$ copies for each of these stored constraints, and $\deg(v)$ copies for each sampled low-degree variable $v$.
Then, for each copy of a constraint adjacent to $v$, we assign one of the copies of $v$ uniformly at random.
Since all the constraints adjacent to $v$ have been stored, this process constructs all the copies of $v$ (and its adjacent constraints) in the reduced \CSP{}.

To handle the high-degree variables, we sample each constraint copy \emph{in the reduced \CSP} with a probability of $2n^{-c}$.
This provides us with a sample of adjacent constraints for each high-degree variable $v$.
The key observation is that the sampled constraints can be used to access $\deg(v) \cdot n^{-c}$ copies corresponding to $v$ in the reduced \CSP{}.
It is at this point that we exploit the specific variant of Trevisan's reduction.
That is, since each constraint copy selects the variable copies uniformly at random, we can simulate sampling $n^{-c}$ of the variable copies as elaborated below. 

Let $F$ denote the set of sampled constraint copies adjacent to a high-degree variable $v$.
To simulate sampling an $n^{-c}$-fraction of the copies in the reduction, 
we start by considering all the copies $v_1, v_2, \ldots, v_{\deg(v)}$ and sampling each of them with a probability of $n^{-c}$, to construct a set of copies $S$.
Then, we subsample $F$ proportional to the size of $S$ to get a set of constraints $F'$.
Finally, each constraint in $F'$ picks a sampled copy of $v$ uniformly at random.
That is, if for a constraint copy $C \in F'$, the original constraint uses $v$ as the $i$-th variable, then $C$ uses a random copy of $v$ in $S$ as its $i$-th variable.
We remark that $F$ was oversampled with a probability of $2n^{-c}$ (rather than $n^{-c}$) to ensure that we have enough edges to supply the $\card{S}$ copies with high probability, even if $\card{S}$ is larger than its mean $\deg(v) \cdot n^{-c}$.
Also, we note that since in (the offline version of) the reduction, each constraint copy chooses one of the variable copies uniformly at random, the instance produced by this process has the same distribution as the offline reduction.
That is to say, going over all constraint copies, picking variable copies uniformly at random, and then sampling some of the variable copies (the offline version), is equivalent to sampling the variable copies, sampling the constraint copies (proportional to the number of sampled variable copies), and then assigning the sampled variables to the sampled constraints uniformly at random (streaming version).

Thus far, we have simulated sampling an $n^{-c}$-fraction of the variables in the reduced \CSP{},
and it remains to produce the final estimate by invoking the \LOCAL{} algorithm.
To produce the final estimate, we sample a uniformly random set of constraints of size $n^{1-c}$ without replacement (this is in addition to the previous samples).
For each sampled constraint, we are able to run the \LOCAL{} algorithm if its $r$-neighborhood has been collected in the previous steps.
This happens with a probability of $n^{-cT}$, where $T$ is the number of dependencies for that constraint. 
Here, the number of dependencies refers to the number of copies in its $r$-neighborhood corresponding to a high-degree variable plus the number of \emph{distinct} low-degree variables for which there is a corresponding copy in the $r$-neighborhood.
To see why the probability of the $r$-neighborhood is $n^{-cT}$,
note that we sample each low-degree variable (in the original \CSP{}) with probability $n^{-c}$,\footnote{More accurately, they are sampled $O(1)$-wise independently.} and if sampled, we replicate all the corresponding copies in the reduced \CSP{}.
In contrast, for high-degree variables, we sample each corresponding copy independently with a probability $n^{-c}$.
The value of $T$ is different for different constraints.
If the $r$-neighborhood of a constraint is collected, we use the \LOCAL{} algorithm to compute its value in the LP, and we also discover the value of $T$.
Then, we factor in the value of the constraint in our estimate while scaling it by $n^{cT}$ to compensate for the cases where the neighborhood is not collected.

We further note that we ensure the degree of the variables fed into the \LOCAL{} algorithm is bounded by a constant.
While Trevisan's reduction, as stated in this section, only bounds the degrees in expectation, we show that removing any constraint copy for which one of the assigned variable copies has a large degree does not perturb the $\maxCSP$ significantly.
As a result, we can upper-bound $T$ for every constraint by a constant.
This results in (1) a higher success rate $n^{-cT}$ in collecting the neighborhoods,
(2) limited dependence between the success events.
Both of these are crucial in obtaining an accurate estimate.

\paragraph{Comparison with \cite{AzarmehrBFS26}.}
As discussed earlier, our work directly implies and generalizes the recent result of \cite{AzarmehrBFS26} for the \maxdicut{} problem. While the two works do share some similarities, such as efficiently simulating suitable local algorithms in a single pass of the streaming setting, there are actually significant differences between the two. In particular, the algorithm of \cite{AzarmehrBFS26} {\em is not} a special case of the algorithm we present in this work. Indeed, the algorithm of \cite{AzarmehrBFS26} heavily relies on the inner workings of a \LOCAL{} subroutine for \maxdicut{}, derived from parallelizing a sequential submodular maximization algorithm. Namely, they exploit the fact that, instead of using the entire $r$-neighborhood, a much sparser subsampled subgraph of the $r$-neighborhood can be used to approximate the output of the \LOCAL{} subroutine, which is specific to \maxdicut{}.
As a result, their analysis requires complicated arguments to address the intricate interaction between low-degree and high-degree vertices, tightly coupled with the \LOCAL{} algorithm.
In contrast, our algorithm uses the \LOCAL{} algorithm of \cite{Yoshida11} as a black box, i.e., the \LOCAL{} algorithm could be potentially replaced with any other \LOCAL{} subroutine that approximates the \BasicLP{}.
Consequently, we are able to decouple the algorithm and abstract away the specificity of the input \CSP{}.
This shifts the focus of the analysis to accessing sampled parts of a reduced input instance, yielding a much simpler and conceptually cleaner generalization of \cite{AzarmehrBFS26}.

\section{Preliminaries \& Basic Setup}

We use $[n]$ to denote the set $\{1, 2, \ldots, n\}$.

\begin{definition}[$\CSP(\mc{F})$]
    For a constant positive integer $k$, alphabet $\Sigma$, and family of predicates $\mc{F} \subseteq \set{f:\Sigma^k \to \set{0,1}}$, define the constraint satisfaction problem $\CSP(\mc{F})$. An instance $I=(V,F)$ of $\CSP(\mc{F})$ is given by a set of $n$ variables $V$ and a set of $m$ constraints $F=(C_1, \dots, C_m)$. Each $C_i=((C_i^1,\dots C_i^k),f_i)$ where, for $t \in [k]$, $C_i^t \in V$ is the $t^{th}$ variable in $C_i$, and $f_i \in \mc{F}$ is a predicate applied over these variables.
\end{definition}

We say that an assignment $\tau : V \to \Sigma$ satisfies constraint $C_i=((C_i^1,\dots,C_i^k),f_i)$ when $f_i(\tau(C_i^1),\dots,\tau(C_i^k)) = 1$. We denote the value of the assignment $\tau$ over instance $I$ as the fraction of constraints satisfied by the assignment:
\[
\val{I}(\tau) = \frac{1}{m}\sum_{i \in [m]} f_i(\tau(C_i^1),\dots,\tau(C_i^k)).
\]

We use $\val{I}$ to denote the maximum $\val{I}(\tau)$ over all assignments $\tau:V\to\Sigma$. We also note that this is equivalent to defining $\val{I}$ as the maximum over all fractional or random assignments of $\tau$. 

Now that we have formally defined $\CSP$s, we also define the $\BasicLP$ relaxation of a $\CSP$. This is a canonical relaxation that has an integrality gap which is heavily tied to the optimality of approximating $\CSP$s in various settings. We give a formal description of the $\BasicLP$ in \Cref{sec:basic-lp}.

\begin{definition}[$\BasicLP$]
    For $\CSP(\mc{F})$, we define the $\BasicLP$ of instance $I=(V,F)$ as an LP relaxation described fully in \Cref{sec:basic-lp}. For an instance $I$, let the optimal value of the $\BasicLP$ for $I$ be $\vallp{I}$. We note that this LP is a relaxation, so for all $I$, $\val{I} \leq \vallp{I}$.
\end{definition}

\begin{definition}[Integrality Gap $\alpha_{\mc{F}}$ of $\BasicLP$] \label{def:int-gap}
    For $\CSP(\mc{F})$, we define $\alpha_{\mc{F}}$ as the integrality gap for the $\BasicLP$. Formally,
    \[
    \alpha_{\mc{F}} = \inf_{I\in\CSP(\mc{F})} \frac{\val{I}}{\vallp{I}}.
    \]
\end{definition}

We note that every $\CSP$ instance $I$ can be represented by a hypergraph where each variable $v$ is a vertex and each constraint $C$ is a hyper-edge adjacent to the $k$ vertices corresponding to its $k$ variables. Because of this, we can borrow some language from graphs to instances of $\CSP$s.

\begin{definition}
    For a variable $v$ in instance $I$ of a $\CSP$, let $\deg_I(v)$ be the degree of $v$ in $I$, or the number of constraints in $I$ adjacent to $v$.
\end{definition}

\begin{definition}
    For a $\CSP(\mc{F})$ and instance $I$ with bounded-degree $B$, let $\mc{N}_{\mc{F}}(r)$ be the set of isomorphism classes of all possible neighborhoods of radius $r$ around a single constraint for any $B$-bounded instance of $\CSP(\mc{F})$. Let $N_I(i,r) \in \mc{N}_{\mc{F}}(r)$ be the specific neighborhood around constraint $C_i$ in instance $I$.
\end{definition}

\begin{definition}
    For a $\CSP$ instance $I=(V,F)$, let $S \subseteq V$. Define the induced instance $I[S] = (S,F[S])$ as a $\CSP$ instance containing only the variables in $S$, and the constraints $F[S]$ adjacent only to variables in $S$. In other words, remove all variables not in $S$ and all of their adjacent constraints.
\end{definition}

We now give a key lemma from \cite{Yoshida11} that gives an additive approximation of $\vallp{I}$ for bounded-degree instances $I$ using a local algorithm. We use the algorithm $\aloc$ as a black-box for our purposes.

\begin{lemma}[\cite{Yoshida11}, Theorem 3.1]\footnote{Though the original theorem appears in \cite{Yoshida11}, we borrow a similar reformulation of this theorem from \cite{FMW25dichotomy}. We also note this theorem was not defined on hypergraphs, but instead on a corresponding bipartite graph. This, however, does not affect the truth of the lemma as stated here.} \label{lem:aloc}
    For any $\CSP(\mc{F})$ with predicate family $\mc{F} \subseteq \set{f:\Sigma^k \to \set{0,1}}$, and constants $B$, $|\Sigma|$, and $\epsilon$. There exists a positive integer $r \leq \exp(\operatorname{poly}(kB|\Sigma|/\epsilon))$ and a deterministic map $\aloc:\mc{N}_{\mc{F}}(r) \to [0,1]$ that satisfy the following:

    Given an instance $I=(V,(C_1, \dots C_m))$ of $\CSP(\mc{F})$ with bounded-degree $B$, the average output of the local algorithm $\aloc(N(i,r))$ for $i \in [m]$ satisfies
    \[
    \vallp{I} -\epsilon \leq \frac{1}{m} \sum_{i \in [m]} \aloc(N(i,r)) \leq \vallp{I} + \epsilon.
    \]
\end{lemma}

\subsection{The Constants}

We use several constants in our analysis that we introduce here for easier access. 
\begin{itemize}
    \item $\epsilon$ is our main error parameter.
    \item $\delta$ is a parameter used for bounding probability. We assume $\delta = O(1/\epsilon)$.
    \item $k$ is the arity of constraints in our $\CSP$, we assume this is constant.
    \item $|\Sigma|$ is the size of our alphabet, we assume this is constant.
    \item $\alpha_{\mc{F}}$ is the integrality gap of the $\BasicLP$. 
    \item $B$ is a parameter used for expected degree in our reduced instances. We take $B = |\Sigma|^{2k}/\epsilon$.
    \item $\rho$ is used alongside $B$ to bound degree in our reduced graph. We take $\rho = \frac{\epsilon}{2k}$.
    \item $r$ is the radius required for our local algorithm to work as needed. $r \leq \exp(\poly(k|\Sigma|B/\epsilon)) \leq \exp(\exp(O(k|\Sigma|/\epsilon)))$.
    \item $\Tmax$ is the maximum number of variables in a local neighborhood of bounded degree $B/\rho$ with radius $r$. $\Tmax \leq \exp(\exp(\exp(O(k|\Sigma|/\epsilon))))$. 
    \item $c$ is a parameter used for sampling variables and constraints. We take $c = \epsilon/\Tmax$. We also assume that $n^{c}$ is an integer for convenience.
    \item $q$ is a parameter used for thresholding between high and low-degree vertices. We pick $q = \epsilon$, but we note that there is a high amount of flexibility with this choice.
    \item Lastly, we assume that $n$ is large enough to overpower all of these constants.
\end{itemize}
\section{Simplifying the Input}
\label{sec:input-simplification}

In this section, we make a series of simplifications to the input,
and show via reductions that these assumptions incur no loss of generality.
For the remainder of the paper, we assume the simplifications hold.

\begin{observation}
    We can assume, without loss of generality, that all the constraints are $k$-ary and nontrivial, where a constraint is trivial if its value is independent of the variable assignment.
\end{observation}
\begin{proof}
    Given an input where the constraints can be $k'$-ary for any $k' \leq k$, we can reduce to a $k$-ary instance by adding $k - k'$ variables to the constraint and disregarding their value.

    Given an algorithm $\mc{A}(I)$ which $\alpha$-approximates the value of $\val{I}$ for instances with non-trivial predicates, we can obtain another algorithm $\mc{A}'(I)$ that achieves the same approximations while allowing trivial predicates.
    This is done by simply running $\mc{A}$ on the non-trivial constraints in the input, and then accounting for the value of the trivial predicates.
    
    More precisely, for an arbitrary input instance $I$, let $m_0$, $m_T$, and $m_F$ denote the number of nontrivial constraints, trivially satisfied (true) constraints, and trivially unsatisfied (false) constraints, respectively.
    Let $m^*$ be the maximum number of nontrivial constraints that can be satisfied at the same time, i.e.\ $\val{I} = (m^* + m_T)/m$.
    We run $\mc{A}$ on the nontrivial constraints, and let $\hat{\mathsf{v}}$ be the reported approximation for $m^*/m_0$. That is,
    $$
    \alpha \frac{m^*}{m_0} \leq \hat{\mathsf{v}} \leq \frac{m^*}{m_0}.
    $$
    Then, we report $(m_0 \hat{\mathsf{v}} + m_T)/m$, as our approximation of $\val{I}$ for the original input.
    It holds that
    $$
    \alpha \val{I} \leq \frac{\alpha m^*+m_T}{m} \leq \frac{m_0 \hat{\mathsf{v}} + m_T}{m} \leq \frac{m^* + m_T}{m} = \val{I}.
    $$
    Therefore, the reported value is an $\alpha$-approximation of $\val{I}$, which concludes the proof.
\end{proof}

\begin{observation}
    We can assume, by incurring an $O(\log n)$ factor in the space complexity, that the number constraints is known to the algorithm up to a constant factor.
\end{observation}
\begin{proof}
    Given that $m = \poly(n)$, i.e.\ the length of the stream is polynomial in the number of variables, we can use $O(\log n)$ copies of the algorithm to \enquote{guess} the value of $m$.
    That is, for each interval $[2^i, 2^{i+1}]$, we run an independent copy of the algorithm in parallel, which assumes the value of $m$ lies in that interval.
    If, at any point, a copy of the algorithm violates the $O(n^{1-c})$ space bound, it is terminated.
    At the end of the stream, the value of $m$ is known to the algorithm, so we can simply report the output of the appropriate copy.
\end{proof}

\begin{claim}[Folklore]\label{claim:linear-edges}
    We can assume, by incurring a $(1 - \epsilon)$ multiplicative factor in the approximation ratio, that the number of constraints is linear in the number of variables, i.e.\ $m = O(n/\epsilon^2)$.
\end{claim}
\begin{proof}
    Let $\lambda = \frac{\epsilon}{2\card{\Sigma}^k}$,
    and $p = (10 \ln \card{\Sigma}) \frac{n}{\lambda^2 m}$.
    Assume $m \geq 10pm \geq \Omega(n/\epsilon^2)$. Otherwise, the claim already holds.
    First, we show that sampling the constraints with probability $p$ preserves the value of each assignment up to an additive error of $\lambda m$ (after scaling by $1/p$).
    Then, we use the fact that $\val{I} \geq \frac{1}{\card{\Sigma}^k}$ to conclude that the value of the optimal assignment is preserved up to a multiplicative error of $(1 \pm \epsilon/2)$.
    As a result, running an $\alpha$-approximation algorithm on the sampled instance yields a $(1 - \epsilon)\alpha$-approximation of the original input.
    Note that the number of sampled constraints (i.e., the size of the sampled input) is at most $10pm = O(n/\epsilon^2)$ with high probability.

    Consider any assignment $\sigma \in \Sigma^n$ of the variables.
    Let $v(\sigma) = \val{I}(\sigma)\cdot m$ denote the number of constraints that are satisfied by this assignment, and let $v_S(\sigma)$ be the number of sampled constraints that are satisfied.
    Observe that $v_S(\sigma)$ is a sum of $v(\sigma)$ $\{0, 1\}$-variables, that take a value of $1$ with probability $p$.
    Therefore, $\Ex[v_S(\sigma)] = p v(\sigma)$, and by the Chernoff bound we have
    $$
    \Pr\left(\abs*{v_S(\sigma)  - pv(\sigma) } > p \lambda m \right) \leq 2\exp\left(-\frac{p^2\lambda^2m^2}{2pv(\sigma) + p\lambda m}\right)
    \leq 2\exp\left(-p\lambda^2 m / 3\right),
    $$
    where the second inequality follows from $pv(\sigma) + p\lambda m \leq 3pm$.
    When the event above happens, i.e.\ $\frac{1}{p}v_S(\sigma)$ approximates $v(\sigma)$, we say that the value of $\sigma$ is preserved.
    Taking union over all possible $q^n$ assignments, the probability that there exists an assignment whose value is not preserved is at most
    $$
    q^n \cdot 2\exp\left(-p\lambda^2m/3\right) \leq e^{-n}.
    $$
    Therefore, the value of every assignment is preserved up to an additive error of $\lambda m$, and so is the optimal value.

    Finally, observe that since there are no trivial constraints, each constraint is satisfied by at least $\frac{1}{\card{\Sigma}^k}$-fraction of the assignments.
    Therefore, there exists an assignment that satisfies at least $\frac{m}{\card{\Sigma}^k}$ constraints, and so does the optimal assignment.
    As a result, a $\lambda m = \frac{\epsilon}{2}\cdot\frac{m}{\card{\Sigma}^k}$ additive error is a $(1 \pm \epsilon/2)$ multiplicative error.
    This concludes the proof.
\end{proof}

\section{Trevisan's Reduction via Offline Two-Tier Subsampling}

In this section, we present a variation of the Trevisan reduction from unbounded to bounded-degree graphs, first described in \cite{trevisan}. Then, we will give an offline, subsampled version of this reduction that will be modified into a streaming algorithm in \Cref{section:streaming}. The goal is to first prove that the offline estimator is accurate, then show that it can be implemented in the streaming setting (i.e., present a streaming algorithm that has the same output as the offline algorithm, with high probability). In combination, this will show that the streaming algorithm is accurate with high probability.

We use a version of the Trevisan Reduction described in \cite{FMW25dichotomy} with slight modification and claim no novelty. 
The reduction works by first making (approximately) $\deg_I(v)$ copies of each variable $v$, and $B$ copies of each constraint. Then, each constraint copy has its adjacent variables randomly distributed among the copies of the corresponding variable. For example, if constraint $C$ is a predicate of variables $v_1, v_2$, and $v_3$. Then each copy of $C$ will have its first variable randomly chosen from the copies of $v_1$, the second from the copies of $v_2$, etc. After this procedure, we will show that the value of the $\CSP$ is preserved and each vertex copy has bounded expected degree.\footnote{We note that some more work is required to get a truly bounded-degree graph, as we cannot union-bound this concentration. This is done differently in \cite{FMW25dichotomy}, but later on, we will instead opt to simply remove any vertex with a degree that is too high, and this will not change our $\CSP$ value much.}

\begin{algorithm}
    \caption{Trevisan Reduction}
    \label{alg:trevisan}
    \KwIn{$\CSP(\mc{F})$ instance $I = (V, (C_1,\dots,C_m))$}
    \KwOut{$\CSP(\mc{F})$ instance $I'$ of constant (expected) degree}
    
    \DontPrintSemicolon
    \SetAlgoSkip{bigskip}
    \SetAlgoInsideSkip{}

    \For {$v \in V$} {
        Adversarially pick a value $\dest(v)$ in the range $(1 \pm \eadv)\cdot \deg_I(v)$. \\
    }
    $V' \gets \set*{(v,j) | v \in V, j \in [\dest(v)]}$\\
    \For{$l \in [B]$}{
        \For{$i \in [m]$}{
            Initialize $C_{(i,l)}$.\\
            \For{$t \in [k]$}{
                Suppose $v$ is the $t^{th}$ variable in $C_i$.\\
                Uniformly pick $j$ from $[\dest(v)]$. \\
                Let the $t^{th}$ variable of $C_{(i,l)}$ be $(v,j)$. \\
            }
        }
    }
    \KwRet{$I' := (V',(C_{(i,l)})_{i \in [m], l \in [B]})$}.

\end{algorithm}

We now prove that the Trevisan Reduction preserves the value of the $\CSP$ instance with high probability. Again, we borrow this proof almost exactly from \cite{FMW25dichotomy} and claim no novelty.

\begin{lemma} \label{lem:trevisan-accurate}
    For a $\CSP(\mc{F})$ instance $I = (V,(C_1, \dots C_m))$, let $I'=(V',(C_{(i,l)})_{i \in [m], l \in [B]})$ be the result of applying \Cref{alg:trevisan} to $I'$. Then, with failure probability at most $\delta$,
    \[
    \val{I} \leq \val{I'} \leq (1+\epsilon) \cdot \val{I}.
    \]
\end{lemma}

\begin{proof}
    First, we can lower-bound $\val{I'}$ by noticing that for any assignment $\tau : V \to \Sigma$, we can create assignment $\tau':V' \to \Sigma$ which assigns the same value to every copy of a variable according to its assignment in $v$. 
    That is, we let $\tau'((v,j)) := \tau(v)$ for all $j$. Then, every copy of a constraint will be satisfied by $\tau'$ if and only if the original constraint is satisfied by $\tau$. Therefore, it holds that $\val{I}(\tau) = \val{I'}(\tau')$, so $\val{I} \leq \val{I'}$.

    Now we upper-bound $\val{I'}$. Fix an assignment $\tau':V'\to \Sigma$. We create an assignment $\tau:V\to\Sigma$ by giving each variable the assignment of one of its copies in $\tau'$, chosen at random.
    That is, we assign $\sigma \in \Sigma$ to variable $v \in V$ with probability 
    \[
    \frac{\abs{\set{j \in [\dest(v)] \mid \tau'((v,j))=\sigma}}}{\dest(v)}.
    \]
    Now, given the fixed assignment $\tau'$ in the reduced instance, we will show that the expected value of the randomized assignment in the original graph is equal to the expected value of the assignment $\val{I'}(\tau')$ over the randomness of the construction of $I'$. Over the randomness of the construction of $I'$, create random variables $X_{(i,l)}$ that take value 1 if $C_{(i,l)}$ is satisfied by the fixed $\tau'$ and 0 otherwise. We see that $\Ex_{I'}[X_{(i,l)}] = \Pr_\tau(C_i \text{ satisfied by } \tau)$. Also, let $X$ be the sum of all $X_{(i,l)}$, then we have
    \[
    \begin{aligned}
        \Ex_{I'}[X] &= \displaystyle \sum_{i \in [m]} \sum_{l \in [B]} \Ex_{I'}[X_{(i,l)}] \\
        &= B\cdot \displaystyle \sum_{i \in [m]} \Pr_\tau(C_i \text{ satisfied by } \tau) = mB\cdot\Ex_\tau[\val{I}(\tau)] \leq mB\cdot\val{I}.
    \end{aligned}
    \]

    Now, we can apply Chernoff (\Cref{lem:add-chernoff}) over the randomness of the construction to show that $X = mB\cdot \val{I'}(\tau')$ is upper-bounded by $(1+\epsilon)\cdot mB\cdot \val{I}$ with high probability. 
    \[
    \Pr_{I'}(X \geq (1+\epsilon)\cdot mB\cdot \val{I}) \leq \exp\paren*{-\frac{((1+\epsilon)\cdot mB\cdot \val{I})^2}{3\cdot (1+\epsilon)\cdot mB\cdot \val{I}}}.
    \]
    Since $\val{I} \geq 1/\card{\Sigma}^k$ and $X = mB\cdot \val{I'}(\tau')$, 
    \[
        \Pr_{I'}(\val{I'}(\tau') \geq (1+\epsilon)\cdot\val{I}) \leq \exp \paren*{-\frac{(1+\epsilon)\cdot mB}{3\card{\Sigma}^k}} \leq \exp \paren*{-\frac{m\card{\Sigma}^k}{3\epsilon}} \leq \frac{\delta}{|\Sigma|^{2mk}}.
    \]
    Thus, for any assignment $\tau'$, the probability that $\val{I'}(\tau') \geq (1+\epsilon)\cdot\val{I}$ is small. Since there are at most $|\Sigma|^{(1+\epsilon)mk}$ possible assignments for $\tau'$, we can union bound to show that
    \[
        \Pr_{I'}(\val{I'} \geq (1+\epsilon)\cdot\val{I}) \leq \delta.
    \]
    This proves the lemma. 
\end{proof}

We also prove that only a small fraction of variables have large degree. 

\begin{lemma} \label{lem:trevisan-bounded-avg}
    For a $\CSP(\mc{F})$ instance $I = (V,(C_1, \dots C_m))$, let $I'=(V',(C_{(i,l)})_{i \in [m], l \in [B]})$ be the result of applying \Cref{alg:trevisan} to $I'$. Then, with failure probability at most $\delta$, the number of constraint copies in $I'$ adjacent to variables with degree greater than $B/\rho$ is at most $mB\epsilon$ (where $\rho = \frac{\epsilon}{2k}$).
\end{lemma}

\begin{proof}
    Consider a variable $v \in V$ with $\deg_I(v)$ adjacent constraints. Then, there are $B\cdot\deg_I(v)$ constraint copies adjacent to the copies of $v$ in $I'$. This implies that the number of variable copies in $I'$ with degree greater than $B/\rho$ is at most $B\cdot \deg_I(v) \cdot \rho/B = \rho\cdot \deg_I(v)$. Given that there are at least $(1-\rho)\cdot \deg_I(v)$ copies of $v$, the probability that any given constraint copy ends up adjacent to a copy of $v$ with degree greater than $B/\rho$ is at most $\rho\cdot\deg_I(v)/((1-\rho)\cdot\deg_I(v)) \leq 3\rho/2$. By union bound, the probability that a constraint copy ends up adjacent to any variable copy with degree greater than $B/\rho$ is at most $3k\rho/2$. Let $X$ be the number of constraint copies adjacent to such variables. Given that $\Ex[X] \leq 3mBk\rho/2$, we can use Chernoff (\Cref{lem:add-chernoff}) to show the following.
    \[
    \Pr(X \geq 2mBk\rho) \leq \exp\paren*{-\frac{(mBk\rho/2)^2}{3(3mBk\rho/2)}} \leq \exp \paren*{-\frac{mBk\rho}{18}} \leq \exp \paren*{-\frac{mB\epsilon}{36}} \leq \delta. 
    \]
    Thus, since $2mBk\rho = mB\epsilon$ the lemma is shown.
\end{proof}

We now define $I'_{\bounded}$ as the $\CSP$ instance attained from replacing all constraints adjacent to variables with degree greater than $B/\rho=2Bk/\epsilon$ with constraints that always fail. From \Cref{lem:trevisan-bounded-avg}, we know there are at most $mB\epsilon$ constraints replaced and $\val{I'} \geq 1/\card{\Sigma}^k$, so we have
\[
    (1-\epsilon\cdot \card{\Sigma}^k)\cdot\val{I'} \leq \val{I'_{\bounded}} \leq \val{I'}.
\]
Combining this with \Cref{lem:trevisan-accurate} yields that with failure probability at most $2\delta$, we have
\begin{equation}\label{eq:I-bdd-val}
    (1-\epsilon\cdot \card{\Sigma}^k)\cdot\val{I} \leq \val{I'_{\bounded}} \leq (1+\epsilon) \cdot \val{I}.
\end{equation}

Next, we define our offline estimator (\Cref{alg:offline-estimator}) which uses \Cref{alg:trevisan} as a subroutine. It starts by performing the Trevisan Reduction while marking variables in the original instance as either \enquote{high-degree} or \enquote{low-degree} based on $\dest(v)$ which is chosen by the adversary. We also reset the estimate to be exact if the variable is \enquote{low-degree} which is a deviation from \Cref{alg:trevisan}, but still falls under the possible behavior of the adversary, so we can still apply \Cref{lem:trevisan-accurate} and \Cref{lem:trevisan-bounded-avg} later in the analysis. After the Trevisan Reduction is performed, we perform a sampling procedure that will later be simulated by our streaming algorithm. The sampling procedure works by keeping copies of high-degree variables independently with probability $n^{-c}$, and for each low-degree variable, we keep all of its copies with probability $n^{-c}$, but forget them otherwise. Each kept variable is stored in a subset $S' \subseteq V'$, and we run an aggregating function (\Cref{alg:aggregate}) on the induced instance $I'[S']$.

\begin{algorithm}
    \caption{Offline Estimator}
    \label{alg:offline-estimator}
    \KwIn{$\CSP(\mc{F})$ instance $I = (V, (C_1,\dots,C_m))$}
    \KwOut{Estimate $\valest$ of $\val{I}$}
    
    \DontPrintSemicolon
    \SetAlgoSkip{bigskip}
    \SetAlgoInsideSkip{}

    \For {$v \in V$} {
        \tcp{Decide if $v$ is a high or low-degree variable.}
        Adversarially pick an integer value $\dest(v)$ in the range $(1 \pm \eadv)\cdot \deg_I(v)$ \\
        \If{$\dest(v) \leq n^{q+c}/B$} {
            $\dest(v) := \deg_I(v)$ \\
            Mark $v$ as low-degree \\
        }
        \Else{
            Mark $v$ as high-degree \\
        }
    }
    \tcp{Perform Trevisan Reduction.}
    $V' \gets \set{(v,j) \mid v \in V, j \in [\dest(v)]}$\\
    \For{$l \in [B]$}{
        \For{$i \in [m]$}{
            Initialize constraint $C_{(i,l)}$ with predicate equal to that of $C_i$\\
            \For{$t \in [k]$}{
                Suppose $v$ is the $t^{th}$ variable in $C_i$\\
                Uniformly pick $j$ from $[\dest(v)]$ \\
                Let the $t^{th}$ variable of $C_{(i,l)}$ be $(v,j)$ \\
            }
        }
    }
    $I' := (V',(C_{(i,l)})_{i \in [m], l \in [B]})$ \\
    \tcp{Sample a random set of variables.}
    Sample $10\Tmax$-wise independent hash function $H \to [n^c]$ \\
    $S' \gets \emptyset$ \\
    \For{$v \in V$ with $v$ low-degree} {
        \tcp{If low-degree parent is sampled by $H$, store all copies.}
        \If{$H(v) = 1$} {
            $S' \gets S' \cup \set{(v,j) \in V' \mid j \in [\dest(v)]}$ \\
        }
    }
    \For{$v \in V$ with $v$ high-degree} {
        \tcp{Store each high-degree copy independently w.p. $n^{-c}$}
        \For{$j \in [\dest(v)]$} {
            \WithProb{$n^{-c}$} {
                $S' \gets S' \cup \set{(v,j)}$ \\
            }
        }
    }
    \tcp{Gather information necessary for aggregation.}
    Uniformly without replacement select $n^{1-c}$ vertices from $\set{C_{(i,l)} \mid i \in [m], l \in [B]}$. Store these in $\C$. \\
    Define $\degs':S' \to \mathbb{N}$ as an map such that $\degs'(v) := \deg_{I'}(v)$ \\
    $\Out \gets \Aggregate(I'[S'], \degs', \C)$\\
    $\valest \gets \frac{\alpha_{\mc{F}}}{1+2\card{\Sigma}^k\cdot\epsilon} \cdot \Out$ \\
    \KwRet{$\valest$} 

\end{algorithm}

Next, we describe the aggregation step. In this step, we are given an induced sub-instance $I'[S']$ of a Trevisan-reduced instance $I'$ of our original instance $I$, along with an array of degrees $\degs':V' \to \mathbb{N}$, and a set $\C$ of constraints from $I'$. 

We begin by forcing our instance to have bounded-degree by removing all constraints adjacent to variables $v$ with $\degs'(v) > B/\rho$. Call the resulting sub-instance $I'[S']_{\bounded}$. We note that $I'[S']_{\bounded}=I'_{\bounded}[S']$ because they have the same variables ($S'$) and the same constraints (all constraints whose adjacent variables are in $S'$ and have $\deg_{I'}(v)=\degs'[v] \leq B/\rho$). After this, we attempt to run $\aloc$ on the $r$-ball around each constraint in $\C$. If the full $r$-ball of $C_{(i,l)}$ in $I'_{\bounded}$ also is in $I'[S']_{\bounded}$, then we can run $\aloc$ on the neighborhood and contribute this to our average. However, we will fail if any variable in the $r$-ball is not in $S'$. To compensate for this, we scale up by the reciprocal of the probability of success so that our estimate is unbiased. 

The check we use to determine if the $r$-ball of a constraint $C_{(i,l)}$ is stored is for all $v \in N_{I'_{\bounded}}(C_{(i,l)}, r-1)$ we have $\deg_{I'[S']_{\bounded}}(v) = \degs'[v]$. This works because if the entire neighborhood is stored, then these degrees will all be equal, and if there is some missing variable that is not in $S'$, then there will be at least one variable in the $(r-1)$-ball with $\deg_{I'[S']_{\bounded}}(v) < \degs'[v]$.

If we succeed, we must accurately determine the probability of success so we can scale properly. To do this, we have to count the number of \enquote{dependencies} our estimate relied on. In other words, each high-degree copy in the $r$-ball appeared in $S'$ independently with probability $n^{-c}$. Each of these counts for a separate \enquote{dependency} for the estimate because they all needed to be sampled for the estimate to be possible. For low-degree variable copies, each distinct parent of a low-degree copy caused all children to appear in $S'$ with probability $n^{-c}$. Each of these low-degree parents count for a separate \enquote{dependency} for the estimate. We see that our entire $r$-ball is in $S'$ if and only if all \enquote{dependencies} succeed. 

Thus, we can define the number of dependencies of $C_{(i,l)}$ as the total number of high-degree copies in $N_{I'_{\bounded}}(C_{(i,l)}, r)$ plus the number of distinct parents of low-degree copies in $N_{I'_{\bounded}}(C_{(i,l)}, r)$. Call this number $T(C_{(i,l)},r)$. Then, since each dependency succeeds independently with probability $n^{-c}$, the probability that the entire neighborhood of $C_{(i,l)}$ is in $S'$ is $n^{-cT(C_{(i,l)},r)}$.

\begin{algorithm}[H]
    \caption{$\Aggregate(I'[S'], \degs', \C)$}
    \label{alg:aggregate}
    \KwIn{$\CSP(\mc{F})$ instance $I'[S']$, set of degrees $\degs':V' \to \mathbb{N}$, and a set of constraints $\C$}
    \KwOut{Estimate for $\vallp{I'}$}
    
    \DontPrintSemicolon
    \SetAlgoSkip{bigskip}
    \SetAlgoInsideSkip{}

    Remove all constraints adjacent to variables $v$ with $\degs'(v) > B/\rho$. Replace them with dummy constraints that always fail. Correspondingly decrease $\degs'(v)$ for all variables adjacent to deleted constraints. \\
    \tcp{Call the resulting instance $I'[S']_{\bounded}$ and the resulting set of degrees $\degs'_{\bounded}[v]$}
    
    $\total \gets 0$ \\
    \For{$C_{(i,l)} \in \C$}{
        \tcp{Check if the entire neighborhood of $C_{(i,l)}$ is in $S'$}
        \If{$C_{(i,l)} \in I'[S']$ and for all $v \in N_{I'_{\bounded}}(C_{(i,l)}, r-1)$ we have $\deg_{I'[S']_{\bounded}}(v) = \degs'_{\bounded}[v]$} {
            \tcp{$T(C_{(i,l)},r) := $ number of dependencies of $C_{(i,l)}$ as defined above.}
            $\total \gets \total + \aloc(N_{I'_{\bounded}}(C_{(i,l)}, r))\cdot n^{cT(C_{(i,l)},r)}$ \\
        }
    }
    $\Out \gets \total/|\C|$ \\
    \KwRet{$\Out$}    
\end{algorithm}

We can now begin to analyze the offline algorithm by fixing the construction of $I'$ such that the good events in \Cref{lem:trevisan-accurate} and \Cref{lem:trevisan-bounded-avg} hold. We now make the following claim about the expectation of our output.

\begin{lemma}\label{lem:offline-ex}
    It holds that
    \[
    \Ex[\Out] = \vallp{I'_{\bounded}}.
    \]
\end{lemma}

\begin{proof}
    For each constraint $C_{(i,j)} \in I'$, we define random variable $X_{(i,l)}$ as $\aloc(N_{I'_{\bounded}}(C_{(i,l)}, r))\cdot n^{cT(C_{(i,l)})}$ if $N_{I'_{\bounded}}(C_{(i,l)}, r) \subseteq S'$, and $0$ otherwise. This represents the contribution to the average of constraint $C_{(i,l)}$. 

    We can see that $\Ex[X_{(i,l)}] = \aloc(N_{I'_{\bounded}}(C_{(i,l)}, r))$ as the event that $N_{I'_{\bounded}}(C_{(i,l)}, r) \subseteq S'$ happens with probability $n^{-cT(C_{(i,l)})}$ as described earlier. If $C_{(i,l)} \notin I'_{\bounded}$, then it was removed and is replaced by a dummy constraint, meaning that $\aloc(N_{I'_{\bounded}}(C_{(i,l)}, r)) = 0$.

    Also create $Y_{(i,l)}$ = $X_{(i,l)}$ if $C_{(i,l)} \in \C$. Then, since the constraints in $|\C|$ are selected uniformly at random without replacement from a total of $mB$ constraints, $\Ex[Y_{(i,l)}] = \frac{|\C|}{mB} \cdot \Ex[X_{(i,l)}] = \frac{|\C|}{mB} \cdot \aloc(N_{I'_{\bounded}}(C_{(i,l)}, r))$. The remainder of the proof comes from simply applying these identities and \Cref{lem:aloc}:
    \begin{align*}
        \Ex[\Out] &= \Ex \bracket*{\frac{1}{|\C|}\displaystyle \sum_{C_{(i,l)} \in \C}  X_{(i,l)}} \\
        &= \Ex \bracket*{\frac{1}{|\C|}\displaystyle \sum_{C_{(i,l)} \in I'}  Y_{(i,l)}} \\
        &= \frac{1}{|\C|}\displaystyle \sum_{i \in [m]} \sum_{l \in [B]} \Ex[Y_{(i,l)}] \\
        &= \frac{1}{|\C|}\displaystyle \sum_{i \in [m]} \sum_{l \in [B]} \frac{|\C|}{mB} \cdot \aloc(N_{I'_{\bounded}}(C_{(i,l)}, r)) \\
        &= \frac{1}{|\C|} \cdot \frac{|\C|}{mB} \cdot mB \cdot \vallp{I'_{\bounded}} = \vallp{I'_{\bounded}}. \qedhere
    \end{align*}
\end{proof}

Now that we know our expectation, we can bound the variance in order to use Chebyshev's inequality (\Cref{lem:chebyshev}) later.

\begin{lemma}\label{lem:offline-var}
    It holds that
    \[
    \Var(\Out) \leq \delta\epsilon^2.
    \]
\end{lemma}

\begin{proof}
    We begin by substituting the value of $\Out$ and using basic properties of variance and covariance.
    \begin{equation}\label{eq:var-decomp}
        \begin{aligned}
            \Var(\Out) &= \Var \paren*{\frac{1}{|\C|}\displaystyle \sum_{C_{(i,l)} \in \C}  X_{(i,l)}} \\
            &= \frac{1}{|\C|^2} \cdot \sum_{C_{(i_1,l_1)} \in \C} \sum_{C_{(i_2,l_2)} \in \C} \Cov(X_{(i_1,l_1)}, X_{(i_2,l_2)}). \\
        \end{aligned}
    \end{equation}
    Now, we upper-bound the number of pairs for which $\Cov(X_{(i_1,l_1)}, X_{(i_2,l_2)})$ is non-zero. We see that two indicator variables $X_{(i_1,l_1)}$ and $X_{(i_2,l_2)}$ are dependent if and only if their probability of success is correlated. This can happen if their neighborhoods share a variable, or if there is some low-degree variable in $I$ containing a copy in each neighborhood. In other words, $X_{(i_1,l_1)}$ and $X_{(i_2,l_2)}$ are correlated if and only if they share a dependency.

    We now consider a single constraint $C_{(i,l)}$ and upper-bound the number of other constraints that share a dependency. Each constraint has at most $\Tmax$ variables in its neighborhood, and each of these variables is in the neighborhood of at most $\Tmax$ other constraints. Thus, there are at most $\Tmax^2$ constraints who's neighborhoods intersect the neighborhood of $C_{(i,l)}$. Next, consider a variable $v$ with low-degree parent in the neighborhood of $C_{(i,l)}$. $v$ can have at most $n^{q+c}/B < n^{q+c}$ siblings, each a member of at most $\Tmax$ different neighborhoods of constraints. Thus, we can upper-bound the number of other constraints sharing a dependency with $C_{(i,l)}$ by $\Tmax^2 + \Tmax^2\cdot n^{q+c} \leq 2\Tmax^2\cdot n^{q+c}$. 

    Also, since $\aloc$ is always between $0$ and $1$, $X_{(i,l)}$ always ranges between $0$ and $n^{c\Tmax}$, meaning the maximum possible covariance between two dependent constraints is $n^{2c\Tmax}$ (due to the $10\Tmax$-wise independence of $H$).

    Combining this with \Cref{eq:var-decomp} gives an upper bound for $\Var(\Out)$.
    \begin{align*}
        \Var(\Out) &\leq \frac{1}{|\C|^2} \cdot (|\C| \cdot 2\Tmax^2\cdot n^{q+c}) \cdot n^{2c\Tmax} \leq \frac{2\Tmax^2\cdot n^{q+c+2c\Tmax}}{n^{1-c}} \leq \delta\epsilon^2. \qedhere
    \end{align*}
\end{proof}

\begin{lemma} \label{lem:offline-out}
    For an instance $I$ of $\CSP(\mc{F})$. Let $\valest$ be as in the output of \Cref{alg:offline-estimator}. Then with probability $1-3\delta$,
    \[
    (\alpha_{\mc{F}}-4\epsilon \card{\Sigma}^k) \cdot \val{I} \leq \valest \leq \cdot \val{I}.
    \]
\end{lemma}

\begin{proof}
    From \Cref{lem:offline-ex} and \Cref{lem:offline-var}, we can apply Chebyshev's inequality (\Cref{lem:chebyshev}).
    \[
    \Pr \paren*{\abs*{\Out - \vallp{I'_{\bounded}}} \geq \epsilon} \leq \frac{\Var(\Out)}{\epsilon^2} \leq \delta.
    \]

    We condition on $\Out$ being within $\epsilon$ of $\vallp{I'_{\bounded}}$ which happens with probability $1-\delta$. From this along with the fact that $\vallp{I'_{\bounded}} \geq \val{I'_\bounded} \geq 1/\card{\Sigma}^k$ and \Cref{eq:I-bdd-val}, we have the following lower bound for $\Out$ which happens with probability $1-3\delta$ by union bound.
    \begin{equation}\label{eq:out-lower}
        \Out \geq \vallp{I'_{\bounded}}-\epsilon \geq (1-\epsilon \card{\Sigma}^k) \cdot \vallp{I'_{\bounded}} \geq (1-\epsilon \card{\Sigma}^k) \cdot \val{I'_{\bounded}} \geq (1-2\epsilon \card{\Sigma}^k) \cdot \val{I}.
    \end{equation}

    The upper-bound is similar and uses the additional fact that $\alpha_{\mc{F}}\cdot\vallp{I'_\bounded} \leq \val{I'_{\bounded}}$ (\Cref{def:int-gap}).
    \begin{equation}\label{eq:out-upper}
        \Out \leq \vallp{I'_{\bounded}} + \epsilon \leq (1+\card{\Sigma}^k\epsilon)\cdot \vallp{I'_{\bounded}} \leq \frac{(1+\card{\Sigma}^k\epsilon)}{\alpha_{\mc{F}}}\cdot \val{I'_{\bounded}} \leq \frac{(1+2\card{\Sigma}^k\epsilon)}{\alpha_{\mc{F}}}\cdot \val{I}.
    \end{equation}

    Now, since $\valest = \frac{\alpha_{\mc{F}}}{1+2\card{\Sigma}^k\cdot\epsilon} \cdot \Out$, we can combine this with \Cref{eq:out-lower} and \Cref{eq:out-upper} to prove the lemma.
    \begin{align*}
        (1-2\epsilon \card{\Sigma}^k) \cdot \frac{\alpha_{\mc{F}}}{1+2\epsilon \card{\Sigma}^k} \cdot \val{I} &\leq \valest \leq \frac{(1+2\card{\Sigma}^k\epsilon)}{\alpha_{\mc{F}}}\cdot \frac{\alpha_{\mc{F}}}{1+2\epsilon \card{\Sigma}^k} \cdot \val{I} \\
        (\alpha_{\mc{F}}-4\epsilon \card{\Sigma}^k) \cdot \val{I} &\leq \valest \leq \val{I}. \qedhere
    \end{align*}
    
\end{proof}
\section{Streaming Implementation}\label{section:streaming}

In this section, we give a streaming implementation of the algorithm from the previous section in one pass and sublinear space in $n$. The algorithm has three major steps: Sketching, reducing, and aggregating. We note that the aggregation step is exactly the same offline and online, and uses the same algorithm (\Cref{alg:aggregate}) as this step is possible in sublinear space. Because of this, the goal of the streaming algorithm is to match the inputs needed to this aggregate function. To do this, we gather all information we need from the stream in a sketch in the body of \Cref{alg:streaming-estimator}, then perform an online version of our offline estimator in \Cref{alg:streaming-reduction}, lastly we end up with a subsampled and reduced version of our original instance $I$ that we can use as input to the aggregate function (\Cref{alg:aggregate}). Again, the behavior we want to mimic is the offline estimator which works by first performing the Trevisan reduction, then subsampling vertices and taking the induced graph on those vertices. The key that makes this possible to implement in the streaming setting is the difference in how we treat low and high-degree vertices.

The sketching portion of the algorithm, \Cref{alg:streaming-estimator}, works by storing several different types of variables and constraints from the stream. The first piece of information we need is a set of variables sampled each with probability $n^{-c}$. This is not possible to do in a streaming setting directly, but we can sample a random $10\Tmax$-wise independent hash function $H:V \to [n^c]$ that can mimic sampling with probability $n^{-c}$ because for each variable $v$, the event that $H(v)=1$ happens with probability $n^{-c}$. As we go through the stream of constraints, we examine each variable $v$ adjacent to that constraint and if $H(v) = 1$, we add that variable to a set $S$. In addition, we keep a degree counter $\degs$ for each variable in $S$.

The second piece of information we need to store are the constraints that will be necessary for the reduction. In the Trevisan reduction, each constraint is copied $B$ times, so we make these constraint copies during the stream by creating constraints $C_{(i,l)}$ that are copies of $C_i$ for $l \in [B]$. These constraints have the same predicate as $C_i$, and the variables adjacent to $C_{(i,l)}$ will correspond to the variables adjacent to $C_i$ similarly to in the offline algorithm. We store all the constraint copies we will need in a set $F$.

First, we keep any constraint adjacent to a variable in $S$. These constraints will be used to distribute over the copies of the low-degree variables in $S$. Second, we store constraint-variable pairs in two sets $G$ and $\Gest$. These constraints will be used to distribute over the copies of high-degree variables. The way this works is for each constraint copy $C_{(i,l)}$ and each variable adjacent to constraint $C_i$ indexed by $t \in [k]$, we store the pair $(C_{(i,l)},t)$ with probability $2n^{-c}$. We do the same process independently for $\Gest$ but with probability $n^{-c}$. Then, for each constraint copy that is a member of some pair in $G$, we store also store it in $F$. The reason for having two separate sets $G$ and $\Gest$ we be expanded upon later, but the reason is it help decouple the randomness of estimating the degrees of high-degree vertices with the randomness of the distributing of its adjacent constraints.

The final piece of information we gather is a uniformly random set of constraint copies of size $n^{1-c}$ that we obtain using reservoir sampling. This is a technique to get a desired number of samples uniformly without replacement from a stream originally described in \cite{vitter1985random}. The technique is described more in \Cref{sec:reservoir-sampling}, but the outcome as it applies to our algorithm is we use two functions $\InitReservoir$ and $\UpdateReservoir$ and at the end of the algorithm we have a sample of $n^{1-c}$ constraint copies selected uniformly at random with replacement stored in a set $\C$. This method of sampling is not strictly required for the algorithm to work, but it makes some of the analysis a bit simpler. 

After all necessary information is stored during the stream, the information is passed along to \Cref{alg:streaming-reduction} which returns a subsampled and reduced version $I'$ of our original instance $I$ that can be passed along to \Cref{alg:aggregate} to aggregate and estimate the $\maxCSP$ value of $I$.

\begin{algorithm}
    \caption{Streaming CSP Estimator}
    \label{alg:streaming-estimator}
    \KwIn{$\CSP(\mc{F})$ instance $I = (V, (C_1,\dots,C_m))$}
    \KwOut{Estimate $\valest$ of $\val{I}$}
    
    \DontPrintSemicolon
    \SetAlgoSkip{bigskip}
    \SetAlgoInsideSkip{}
    Sample $10\Tmax$-wise independent hash function $H: V \to [n^c]$. \\
    \tcp{$S$ is a set of variables, $F$ is a set of constraint copies, $G, \Gest$ are sets of constraint-variable pairs, and $\degs$ is a map from $S \to \mathbb{N}$.}
    $S \gets \emptyset, F \gets \emptyset, G \gets \emptyset, \Gest \gets \emptyset, \degs \gets []$\\
    \tcp{$\C$ is a reservoir of constraints defined by the data-structure in \Cref{sec:reservoir-sampling}.}
    $\C \gets \InitReservoir(n^{1-c})$ \\
    \For{$i \in [m]$} {
        \For{$t \in [k]$} {
            \If{$H(C_i^t) = 1$ and $C_i^t \notin S$} {
                \tcp{Hash succeeds with probability $n^{-c}$.}
                $S \gets S \cup \set{C_i^t}$ \\
                $\degs[C_i^t] \gets 0$ \\
            }
        }
        \For{$l \in [B]$}{
            \tcp{Loop through each copy of the constraint.}
            Define constraint $C_{(i,l)}$ with predicate equal to that of $C_i$, but wait to set its variables. \\ 
            $\UpdateReservoir(\C, C_{(i,l)})$ \\
            \For{$t \in [k]$}{
                \WithProb{$2n^{-c}$}{
                    $G \gets G \cup (C_{(i,l)},t)$ \\
                }
                \WithProb{$n^{-c}$}{
                    $\Gest \gets \Gest \cup (C_{(i,l)},t)$ \\
                }
                \If{$C_i^t \in S$} {
                    $\degs[C_i^t] \gets \degs[C_i^t] + 1$ \\
                }
            }
            \tcp{Keep the constraint if adjacent to a variable in $S$ or sampled in $G$}
            \If{$C_i^t \in S$ or $(C_{(i,l)},t) \in G$ for any $t \in [k]$} {
                $F \gets F \cup C_{(i,l)}$ \\
            }
        }
    }

    $(I', \degs') \gets \StreamingReduction(S, F, G, \Gest, \degs, \C)$\\ 
    $\Out \gets \Aggregate(I', \degs', \C)$ \\
    $\valest \gets \frac{\alpha_{\mc{F}}}{1+2\card{\Sigma}^k\cdot\epsilon} \cdot \Out$ \\
    \KwRet{$\valest$} 

\end{algorithm}

After the sketching is complete, we must mimic the subsampled Trevisan reduction in \Cref{alg:offline-estimator} using the information we have from the sketch. To do this, we first define some notation. Let $\deg_{G}(v) = |\{(C_{(i,l)},t) \in G \mid C_i^t = v\}|$ and let $G_v$ be the set of variables referred to in $G$, i.e. $v \in G_v$ when there exists $(C_{(i,l)},t) \in G$ with $C_i^t = v$. We also define the same for $\Gest$. 

To perform the subsampled reduction, we first use the threshold of $\deg_{\Gest}(v) = n^q$ to determine if a variable is high-degree or low-degree. If $\deg_{\Gest}(v) \leq n^q$ then our variable is low-degree, otherwise it is high-degree. We use $\deg_{\Gest}(v)$ because it is a good proxy for the degree of $v$, in fact for high-degree variables, we will show later that $\dest(v) = \deg_{\Gest}(v)\cdot n^c/B$ is within an $\epsilon$ factor of $\deg_I(v)$. This degree estimate and cutoff between low and high-degree variables is the only purpose of $\Gest$, as it makes it so that everything we do with the constraints in $G$ is independent of this estimate. 

Now that we have a way to determine if a variable is low or high-degree, we describe how the reduction is performed in each case. For low-degree variables, we only deal with variables in $S$. For each such variable, we have that variable and all neighboring constraint copies stored. Because of this, we can simply perform the reduction as we would offline where we create $\degs[v]$ vertex copies labeled $(v,j)$ for $j$ from $1$ to $\degs[v]$. Then, we take each constraint copy $C_{(i,l)}$ and index $t \in [k]$ such that $C_i^t = v$ and set $C_{(i,l)}^t$ to a uniformly random variable copy $(v,j)$. This simulates the offline algorithm as, offline, the reduction happens first, then each copy of a low-degree variable is \enquote{sampled} only if the parent is sampled. Online, we first sample the parent and store the constraints, then distribute the adjacent constraints over the variable copies. 

For high-degree variables, we estimate the degree of variable $v$ by $\dest(v) = \deg_{\Gest}(v)\cdot n^c/B$. Then, we simulate the process of sampling variable copies with probability $n^{-c}$ by looping from $1$ to $\dest(v)$ and creating a variable copy with this probability. Then, after the variable copies are created, we need to distribute the constraints over these variables. Offline, this is done by distributing first in the reduction, then subsampling. Online, we subsample variable copies first, then we re-sample the same constraint-variable pairs in $G$ that were previously sampled with probability $2n^{-c}$ again with probability $j/(\dest(v)\cdot 2n^{-c})$ where $j$ is the number of succeeding variable copies from the previous step. This gives an overall probability of success for each constraint copy of $j/\dest(v)$. Then, if the constraint copy succeeds, it is randomly distributed over one of the corresponding variable copies.

After constraint copies are distributed, the final step is to remove all constraint copies in $F$ that have unassigned variables. This way we only keep constraints that have all adjacent variable copies sampled, similarly to the offline algorithm which takes the induced instance over sampled variables. 

\begin{algorithm}
    \caption{Streaming Reduction}
    \label{alg:streaming-reduction}
    \KwIn{Stored information $S, F, G, \Gest, \degs$, and $\C$ from $\StreamingEstimator$}
    \KwOut{Reduced instance $I'$ and a map of degrees $\degs'$.}
    
    \DontPrintSemicolon
    \SetAlgoSkip{bigskip}
    \SetAlgoInsideSkip{}

    $V' \gets \emptyset$ \\
    $\degs' \gets []$ (hashmap) \\
    \tcp{Low-degree variables}
    \For{$v \in S$ with $\deg_{\Gest}(v) \leq n^q$}{
        Create $(v,j)$ for $j \in \degs[v]$ \\
        $V' \gets V' \cup \set{(v,j) \mid j \in [\degs[v]]}$ \\
        $\degs'[(v,j)] \gets 0$ for $j \in \degs[v]$ \\
        \tcp{Distribute adjacent constraints over low-degree variable copies.}
        \For{$C_{(i,l)} \in F, t \in k$ such that $C_i^t=v$}{
            Select $j$ randomly from $1$ to $\degs[v]$ \\
            $C_{(i,l)}^t \gets (v,j)$ \\
            $\degs'[(v,j)] \gets \degs'[(v,j)] + 1$ \\
        }
    }
    \tcp{High-degree variables}
    \For{$v\in G_v$ with $\deg_{\Gest}(v) > n^q$}{
        $\dest(v) \gets \deg_{\Gest}(v)\cdot n^c/B$ \\
        $j \gets 0$ \\
        \tcp{Simulate $n^{-c}$ fraction of copies being stored.}
        \For{$\_ \in [\dest(v)]$}{
            \WithProb{$n^{-c}$}{
                $j \gets j + 1$ \\
                Create $(v,j)$ \\
                $V' \gets V' \cup \set{(v,j)}$ \\
                $\degs'[(v,j)] \gets 0$ \\
            }
        }
        \tcp{Conditionally distribute adjacent constraints over succeeding copies.}
        \For{$(C_{(i,l)},t) \in G$ such that $C_i^t=v$}{
            \WithProb{$j/(\dest(v)\cdot 2n^{-c})$}{
                Select $j'$ randomly from $1$ to $j$. \\
                $C_{(i,l)}^t \gets (v,j')$ \\
                $\degs'[(v,j')] \gets \degs'[(v,j')] + 1$ \\
            }
        }
    }
    Remove all constraints from $F$ that have unassigned variables. \\
    \KwRet{$(I'=(V', F), \degs')$}
\end{algorithm}

\subsection{Analysis of the Streaming Algorithm}

We begin our analysis with the space complexity. It is clear that $\C$ requires space $O(n^{1-c})$ always, and $S$ and $\degs$ will each require space $O(n^{1-c})$ w.h.p. as each variable is in $S$ with probability $n^{-c}$.\footnote{Full independence between the events of these events is not necessary here, pairwise independence guaranteed by our hash function is sufficient for this claim via Chebyshev. We gloss over the proof of this as we prove an even stronger claim bounding the space of $F$.} $G$ and $\Gest$ will each require space $O(kBmn^{-c}) =O(n^{1-c})$ w.h.p. as each constraint has $kB$ independent chances to be stored in each of $G$ and $\Gest$ with probability $n^{-c}$ and $m$ is linear in $n$ by \Cref{claim:linear-edges}. 

The final structures we need to bound are $V'$, $\degs'$, and $F$. There are two contributions to the size of $V'$ (and correspondingly $\degs'$), copies of low-degree variables and copies of high-degree variables. We first count the number of copies of high-degree variables. A copy of a high-degree variable $v$ is added to $V'$ with probability $n^{-c}$ independently, a total of $\dest(v)$ times. We will show later (\Cref{claim:couple-G-adv}) that $\dest(v) \leq (1+\epsilon)\cdot \deg_I(v)$ for all high-degree $v$ with high probability and the sum of degrees over all variables is at most $mk$. Thus, with high probability, the number of high-degree copies in $V'$ is $O(mkn^{-c}) = O(n^{1-c})$ (ignoring constant factors). The number of copies of low-degree vertices in $V'$ is strictly less than the size of $F$, so it suffices to bound the size of $F$ asymptotically. This is because the number of low-degree copies in $V'$ is the sum of $\degs[v] = \deg_I(v)$ for each low-degree $v$ in $S$, and $F$ contains $B$ copies of every constraint adjacent to any variable in $S$.

To bound the size of $F$, we use a technique similar to the proof of \Cref{lem:offline-var}. We create indicator variables $X_{(i,l)}$ for each constraint $C_{(i,l)}$ which take the value $1$ if $C_{(i,l)}$ is added to $F$ in the sketch. This happens if, for any $t \in [k]$, $C_i^t \in S$ or $(C_{(i,l)},t) \in G$. The former of these probabilities happens with probability $n^{-c}$ and the latter with probability $2n^{-c}$. Thus, by union bound, $\Ex[X_{(i,l)}] \leq 3kn^{-c}$ for each constraint $C_{(i,l)}$. Let $X$ be the sum of these indicator variables, so $\Ex[X] \leq 3mBkn^{-c}$. Now, we bound variance by expanding the sum of indicator variables to a sum of covariances.
\[
\Var(X) = \Var \paren*{\sum_{i \in [m], l \in [B]} X_{(i,l)}} = \sum_{i_1 \in [m], l_1 \in [B]} \sum_{i_2 \in [m], l_2 \in [B]} \Cov \paren*{X_{(i_1,l_1)},X_{(i_2,l_2)}}.
\]

Since these are indicator variables with expectation bounded by $3kn^{-c}$, their covariance must also be bounded by $3kn^{-c}$. We can then use a crude bound that all pairs of constraints are covaried, giving us $\Var(X) \leq (mB)^2\cdot 3kn^{-c}$. Then we can use a one-sided version of Chebyshev's inequality, but with a loose bound of $2\cdot 3mBkn^{-c/3} \geq 2\cdot 3mBkn^{-c} \geq 2\Ex[X]$. 
\[
\Pr(X \geq 2\cdot 3mBkn^{-c/3}) \leq \frac{\Var(X)}{(3mBkn^{-c/3})^2} \leq \frac{(mB)^2\cdot 3kn^{-c}}{(3mBkn^{-c/3})^2} \leq \frac{1}{3kn^{c/3}}.
\]

Thus, with high probability, we can see that $|F| \leq 2\cdot 3mBkn^{-c/3} = O(n^{1-c/3})$ (ignoring constant factors and using the fact that $m = O(n/\epsilon^2)$). Putting everything together, we see that with high probability, our total space is $O(n^{1-c/3}) = O(n^{1-\Omega_\epsilon(1)})$, and we say that we fail with probability at most $\delta$.

We also note that aside from failing from too much space used, there is one condition in \Cref{alg:streaming-reduction} that can cause us to fail, namely the re-sampling of constraints in $G$ with probability $j/(\dest(v)\cdot 2Bn^{-c})$. If this probability is greater than $1$ then it is not possible to sample with that probability, so we say that the algorithm terminates in this case. The good thing is this happens with low probability because $j/(\dest(v)\cdot 2n^{-c}) > 1$ when $j > (\deg_{\Gest}(v)\cdot n^c/B) \cdot 2n^{-c} = 2\deg_{\Gest}(v)/B$. We can write $j$ as a sum of $\deg_{\Gest}(v)\cdot n^c/B$ independent Bernoulli random variables each with probability $n^{-c}$. Then since $\Ex[j] = \deg_{\Gest}(v)/B$, we can use \Cref{lem:mult-chernoff} to see the following:
\[
\Pr(j > 2\deg_{\Gest}(v)/B) \leq \exp (-\deg_{\Gest}(v)/3B).
\]

Since $\deg_{\Gest}(v) > n^q$, this probability is less than $\delta/mBk$. Union bounding over all possible $(C_{(i,l)},t)$ gives an upper bound of $\delta$ for the probability that any constraint-variable pair fails in this way.

\subsection{Coupling between online and offline algorithms}

We now give a coupling between the result of \Cref{alg:streaming-estimator} and \Cref{alg:offline-estimator} when applied to the same input. Specifically, we can couple the randomness in the two algorithms such that their results are equal with high-probability. To do this, we must list all sources of randomness from both algorithms, define how they are coupled, and show that this yields the same outcome with high probability. It also should be noted that the result of the streaming algorithm is independent of the stream order.\footnote{Formally, we claim that the distribution of possible outputs of the streaming algorithm is independent of the stream order. This is true because the output depends only on the sketch $(S,F,G,\Gest,\degs,\C)$, and this sketch is also independent of the order. $S$ and $\degs$ are determined only by the hash function and the structure of the input, $\C$ is independent due to the nature of reservoir sampling, and $F,G,$ and $\Gest$ are unordered sets where each constraint or constraint-variable pair is added independently.} This is important as it implies that the output of the online and offline algorithms are dependent only on the structure of the input and the internal randomness.

Occasionally to avoid ambiguity we mark variables and constants that exist in both the online and offline algorithms with a corresponding \enquote{on} or \enquote{off} subscript (Ex: $\dest(v)_{\on}, \dest(v)_{\off}$). This is omitted when the focus is clearly on one algorithm or the other. We seek to prove the following lemma:

\begin{lemma}\label{lem:coupling}
    For a given input $I = (V, (C_1, \dots, C_m))$ to both \Cref{alg:offline-estimator} and \Cref{alg:streaming-estimator}, there is a coupling between the distributions of their outputs $\Out_{\off}$ and $\Out_{\on}$ such that $\Out_{\off} =\Out_{\on}$ with probability at greater than $1-3\delta$.
\end{lemma}

\begin{proof}
    We start by taking inventory of the sources of randomness in each algorithm.
    
    Online, in \Cref{alg:streaming-estimator}, the randomness comes from sampling the hash function, each reservoir update to $\C$, and the constraint sampling that happens for $G$ and $\Gest$ respectively. Later, in \Cref{alg:streaming-reduction}, the randomness for low-degree variables comes from the constraint distribution over variable copies, and for high-degree variables there is randomness in the sampling of its simulated copies as well as the distribution of constraints over those copies. There is no randomness in the aggregation step. 

    Offline, in \Cref{alg:offline-estimator}, the randomness comes from the constraint distribution, the sampling of the hash function, the sampling of high-degree variable copies, and the final step uniformly selecting the constraints in $\C$ without replacement.
    
    We also give special attention to the adversary in \Cref{alg:offline-estimator} that determines the value of $\dest(v)$. Any property that holds for all decisions of the adversary, will also hold when the adversary is replaced by a randomized process that abides by the same rules. Because of this, we can also couple the decisions of the adversary with the randomness of the online algorithm. 

    We now describe the coupling in steps. In each step, we will reveal some randomness from both algorithms, couple that randomness, then fix the result. Each time we will get desired properties with certainty or with high probability that will help us prove the final equality $\Out_{\off} =\Out_{\on}$ holds with high-probability for some coupling.
    
    \begin{claim}\label{claim:couple-HC}
        For a given input $I = (V, (C_1, \dots, C_m))$ to both \Cref{alg:offline-estimator} and \Cref{alg:streaming-estimator}, there is a coupling between the randomness of $H_{\on},H_{\off}$ and $\C_{\on},\C_{\off}$ in both algorithms independent of all other internal randomness with the following properties:
        \begin{itemize}
            \item For $v \in V$, $H_{\on}(v) = H_{\off}(v)$.
            \item $\C_{\on}=\C_{\off}$.
        \end{itemize}
    \end{claim}
    \begin{proof}
        The hash function $H$ is sampled exactly the same way in both algorithms, and the resulting set of constraints in $\C$ will always be a uniform without-replacement sample of size $n^{1-c}$ from $\set{C_{i,l} \mid i \in [m], l \in [B]}$. Offline, this is made explicit, while in the online algorithm this ends up being the result due to of the nature of reservoir sampling. Since the hash function ($H$) and the reservoir sampling ($\C$) act the same in both algorithms, we can perfectly couple their randomness such that for $v \in V$, $H_{\on}(v) = H_{\off}(v)$ and $\C_{\on}=\C_{\off}$, proving the claim.
    \end{proof}

    We now fix the randomness of $H$ and $\C$ under the assumption of the two properties in the previous claim. Next, we couple the randomness of $\Gest$ in the online algorithm with the behavior of the adversary in the offline algorithm with the following claim.

    \begin{claim}\label{claim:couple-G-adv}
        For a given input $I = (V, (C_1, \dots, C_m))$ to both \Cref{alg:offline-estimator} and \Cref{alg:streaming-estimator} and the randomness previously fixed in \Cref{claim:couple-HC}, there is a coupling between the randomness of $\Gest$ online and the adversary offline independent of all other remaining internal randomness with the following properties that all hold with probability at least $1-\delta$.
        \begin{itemize}
            \item For all $v \in V$, in both algorithms, $v$ is determined to be either low-degree or high-degree, and this determination agrees between the two algorithms.
            \item $v \in V$ is high-degree (in both algorithms) if and only if $\deg_{\Gest}(v) > n^{q}$.
            \item For all $v \in S$ (online) and $v$ low-degree, $\dest(v)_{\off} = \degs[v] = \deg_I(v)$.
            \item For all high-degree $v \in V$, $\dest(v)_{\off} = \dest(v)_{\on} \in (1 \pm \epsilon)\cdot \deg_I(v)$.
        \end{itemize}
    \end{claim}
    \begin{proof}
        In the offline algorithm, the algorithm picks a $\dest(v)_{\off} \in (1 \pm \epsilon) \cdot \deg_I(v)$, then if $\dest(v)_{\off} \leq n^{q+c}/B$, it resets $\dest(v)_{\off} := \deg_I(v)$. As mentioned earlier, we couple an adversary by replacing it with a randomized process and then coupling the randomized process. This preserves all properties that hold for any behavior of the adversary while allowing us to couple properly. We define the random process as $\dest(v)_{\off} := \deg_{\Gest}(v) \cdot n^c/B$ if $\deg_{\Gest}(v) \cdot n^c/B \in (1 \pm \epsilon) \cdot \deg_I(v)$ and fail otherwise (pick any value in the range, we will ignore this case in the analysis). 

        First, we show that we fail with low probability. Take a variable $v$ with $\deg_I(v) > n^q$. Create an indicator variable for each pair $(C_{(i_1,l_1)},t_1) , \dots, (C_{(i_{B\deg_I(v)},l_{B\deg_I(v)})}, t_{B\deg_I(v)}))$ with $C_{i_j}^{t_j} = v$ and call them $X_1, \dots, X_{B\deg_I(v)}$. Let the variable $X_j$ be equal to $1$ when $(C_{(i_j,l_j)},t_j) \in \Gest$ and $0$ otherwise. The sum of these random variables is $\deg_{\Gest}(v)$. Each variable is independent with $\Ex[X_j] = n^{-c}$. We then use Chernoff (\Cref{lem:mult-chernoff}) to see:
        \[
        \Pr \paren*{\abs*{\deg_{\Gest}(v) - B\cdot\deg_I(v)\cdot n^{-c}} \geq \epsilon B\cdot \deg_I(v)\cdot n^{-c}} \leq 2\exp \paren*{-\frac{\epsilon^2 B \cdot\deg_I(v)\cdot n^{-c}}{3}}.
        \]

        Union bounding over all variables with $\deg_I(v) > n^q$ implies that the probability of a single failing variable is at most $2n\exp \paren*{-\epsilon^2 B\cdot n^{q-c}/3} \leq \delta$. We condition on none of the variables failing for the remainder of the proof.

        We see that offline, $v$ is high-degree if $\dest(v)_\off = \deg_{\Gest}(v) \cdot n^c > n^{q+c}$, or $\deg_{\Gest}(v) > n^q$. This is also the same condition for the online algorithm. In both cases, if $\deg_{\Gest}(v) \leq n^q$ then $v$ is low-degree. In addition, if $v \in S$ and $v$ is low-degree, then $\dest(v)_{\off} = \degs[v] = \deg_I(v)$, and if $v$ is high-degree then $\abs*{\deg_{\Gest}(v) - \deg_I(v)\cdot n^{-c}} \leq \epsilon\cdot \deg_I(v)\cdot n^{-c}$, so multiplying by $n^c$ gives $\dest(v)_{\off} = \dest(v)_{\on} \in (1 \pm \epsilon)\cdot \deg_I(v)$. Showing all 4 properties and concluding the proof of \Cref{claim:couple-G-adv}
    \end{proof}
    
    For the remainder of the proof, we now condition on the coupling in \Cref{claim:couple-G-adv} succeeding and fix the randomness of $\Gest$ and the offline adversary. We also note the following corollary which is necessary in bounding the space complexity in the previous section. 

    Consider a variable $v$ \enquote{sampled} if either $v \in S$ and $v$ is low-degree, or $v$ is high-degree. Consider a copy of a variable $(v,j)$ \enquote{Trevisan sampled} in the streaming algorithm if $(v,j) \in V'_{\on}$, and $(v,j)$ is \enquote{Trevisan sampled} in the offline algorithm if $(v,j) \in S'_{\off}$. We will show there is a one to one correspondence between variable copies and their adjacent constraints for some coupling. 

    We first discuss low-degree variables, in the online algorithm, a copy $(v,j)$ of a low-degree variable $v$ is in $V'_{on}$ exactly when $v \in S$, or $v$ is \enquote{sampled}. In the offline algorithm, a copy $(v,j)$ of low-degree variable $v$ is in $I'[S']$ exactly when $(v,j) \in S'$, so $H(v)=1$ and $v \in S$. This implies that copies of low-degree variables are \enquote{Trevisan sampled} exactly when their parent variable is \enquote{sampled} in both algorithms. Furthermore, each \enquote{sampled} low-degree $v$ has the same number of variable copies in each algorithm as $\degs[v] = \dest(v)_{\off}$ from \Cref{claim:couple-G-adv}. In addition to this, in both algorithms, each adjacent constraint $C_i$ to $v$ has $B$ copies that are uniformly and independently distributed over the copies of $v$. Because of this, we can couple these distributions and assume for the remainder of the proof that for any \enquote{Trevisan sampled} variable copy $(v,j)$, constraint $C_{(i,l)}$ and $t \in [k]$, $C_{(i,l)}^t = (v,j)$ online exactly when $C_{(i,l)}^t = (v,j)$ offline.

    Next, we discuss high-degree variables. From \Cref{claim:couple-G-adv}, we know that for high-degree variable $v$, we have $dest(v)_{\off} = \dest(v)_{on}$ (as these are the same we can simply refer to this value as $\dest(v)$). Because of this, both algorithms sample the exact same number of variable copies independently with probability $n^{-c}$. We can then couple this sampling so that for some $j \in [\dest(v)]$, the $j^{th}$ sample succeeds offline if and only if it succeeds online. Despite the coupling, the \enquote{Trevisan sampled} variables likely will have different labels because in the online algorithm, if we have $j$ \enquote{Trevisan sampled} variable copies for a vertex $v$, they will be labeled $(v,1), \dots, (v,j)$. However, offline, they will be labeled in their original range of $1$ to $\dest(v)$. We can adjust for this with a map $M_v:[j]\to[\dest(v)]$ such that $M_v(j')$ is the index of the $j'^{th}$ \enquote{Trevisan sampled} variable copy of $v$. This mapping is chosen such that variable copy $(v,j')$ online had its sampling correlated with $(v,M_v(j'))$ offline. After fixing the randomness of this part of the coupling, the only randomness left is the randomness of $G$ and the distribution of constraints over high-degree vertices. 

    We examine the process of high-degree constraint distribution in both algorithms. Offline, we simply distribute constraints uniformly at random over all variable copies, then do the offline sampling process which was fixed above. Online, we first sample a given constraint $C_{(i.l)}$ adjacent to high-degree vertex $v$ with probability $2n^{-c}$, then we count degrees and do the online sampling process fixed above, then we sample again with probability $j/(\dest(v)\cdot2n^{-c})$ where $j$ is the number of \enquote{Trevisan sampled} variable copies of $v$. Lastly, we pick a \enquote{Trevisan sampled} variable copy of $v$ uniformly at random to assign our constraint. 

    First, we can see that the value of $j$ has been fixed above, so for each constraint, we sample first with probability $2n^{-c}$, and if we succeed, we sample again with probability $j/(\dest(v)\cdot2n^{-c})$. This is the same as sampling from the start with probability $j/\dest(v)$. In addition, sampling with probability $j/\dest(v)$, then distributing over $j$ variable copies is the same as first distributing over $\dest(v)$ copies and then keeping only $j$ out of those $\dest(v)$ copies. We see that in both cases we end up with $j$ succeeding variables and each constraint is assigned to any one of these variables with probability $1/\dest(v)$. Since these distributions are identical, we can couple them so they have identical outcomes. Similar to before, we assume that for any \enquote{Trevisan sampled} variable copies $(v,j)_{\on}, (v,M_v(j))_{\off}$, constraint $C_{(i,l)}$ and $t \in [k]$, $C_{(i,l)}^t = (v,j)$ online exactly when $C_{(i,l)}^t = (v,M_v(j))$ offline. 

    This gives an isomorphism between $I'=(V',F)$ online and $I'[S']$ offline. The map is simply $(v,j) \to (v,j)$ for low-degree $v$, and $(v,j) \to (v,M_v(j))$ online. Constraints are simply mapped by $C_{(i,l)} \to C_{(i,l)}$. We know this is an isomorphism because the variable mapping is a bijection, the constraint mapping is a bijection, and the isomorphism preserves adjacency.
    
    This also implies that $\degs'_{\on} = \degs'_{\off}$. Since we already know that $\C_{\on}=\C_{\off}$ from \Cref{claim:couple-HC}, all inputs to the aggregation function are equivalent. Since the aggregation function works the same in both cases, we have $\Out_{\on} = \Out_{\off}$. For this to be true, we had to condition on having the proper space bound, not terminating because $j/(\dest(v)\cdot 2n^c) > 1$, and the probability bound in \Cref{claim:couple-G-adv}. Each of these events fail with probability at most $\delta$ giving us a total failure probability of at most $3\delta$, proving \Cref{lem:coupling}.
\end{proof}

\begin{theorem}[Restating of \Cref{thm:main}]
    For a $\CSP(\mc{F})$, let $\alpha_{\mc{F}}$ denote the integrality gap of the corresponding \BasicLP{}. There is a single-pass streaming algorithm that, with probability $1-\delta$, $(\alpha_{\mc{F}} - \epsilon)$-approximates $\maxCSP{(\mc{F})}$ using $O(n^{1-c})$ space. Where $c \geq 1/\exp(\exp(\exp(\poly(k{|\Sigma|}/\epsilon)))$. 
\end{theorem}

\begin{proof}
    We prove our final theorem by combining \Cref{lem:coupling} with \Cref{lem:offline-out}. If $\Out_{\on} = \Out_{\off}$ with probability $1-3\delta$, then by union bound we can see that with probability $1-6\delta$, the exact statement of \Cref{lem:offline-out} applies to the online algorithm as well, namely
    \[
    (\alpha_{\mc{F}}-4\epsilon \card{\Sigma}^k) \cdot \val{I} \leq \valest \leq \val{I}.
    \]
    The theorem follows from a reparameterization of the constants. It is worth noting that we keep $c \geq 1/\exp(\exp(\exp(O(k{|\Sigma|}/\epsilon))))$, so this does not another layer to the exponent. We also note that if we consider $k$ and $|\Sigma|$ to be constant, then we get a bound of $c \geq 1/\exp(\exp(\poly(1/\epsilon))$.
\end{proof}

\printbibliography

\appendix
\section{General Lemmas About Statistics}

\begin{proposition}[Chebyshev's Inequality]\label{lem:chebyshev}
    For any random variable $X$ with mean $\mu$ and variance $\sigma^2$ and constant $t > 0$, it holds that
    \[
    \Pr \paren*{\abs*{X-\mu} \geq t} \leq \frac{\sigma^2}{t^2}.
    \]
\end{proposition}

\begin{proposition}[Multiplicative Chernoff Bound]\label{lem:mult-chernoff}
Let $X_1, \dots, X_n$ be independent random variables taking values in $[0,1]$, 
and let $X = \sum_{i=1}^n X_i$ with $\mu = \mathbb{E}[X]$. 
For any $0 < \delta < 1$,
\[
\Pr \paren*{\abs{X - \mu} \geq \delta \mu} \leq 2 \exp\paren*{-\frac{\delta^2}{3} \mu}.
\]
\end{proposition}

\begin{proposition}[Additive Chernoff]\label{lem:add-chernoff}
Let $X_1,\dots,X_N$ be independent Bernoulli random variables and set
$X:=\sum_{i=1}^N X_i$ with $\mu:=\mathbb{E}[X]$. Then for any $\lambda>0$,
\[
\Pr\!\left[\,|X-\mu|\ge \lambda\,\right]
\;\le\; 2\exp\!\left(-\frac{\lambda^2}{2\mu + \lambda}\right).
\]
\end{proposition}

\section{$\BasicLP$ for $\CSP$s}\label{sec:basic-lp}

In this section, we formally describe the $\BasicLP$ that we are approximating. Though the details of the LP can be ignored for the purposes of this paper, we describe it here and give some intuition behind the formulation as context. We use the formulation from \cite{FMW25dichotomy}.

For an instance $I=(V,(C_1, \dots, C_m))$ of $\CSP(\mc{F})$, let constraint $C_i$ have variables $v_{i,1}, \dots v_{i,k} \in V$ and predicate $f_i \in \mc{F}$. We define the $\BasicLP$ relaxation for instance $I$ as follows with variables $(x_{v,\sigma})_{v \in V,\sigma \in \Sigma}$ and $(z_{i,b})_{i \in [m],b\in \Sigma^k}$.

\[
\begin{alignedat}{3}
\text{maximize} \quad 
& \frac{1}{m}\sum_{i \in [m]} \sum_{b \in \Sigma^k} f_i(b)\cdot z_{i,b} 
& & & \\
\text{s.t.} \quad 
& \sum_{\sigma \in \Sigma} x_{v,\sigma} = 1 
&\quad & \forall v \in V \\
& \sum_{b \in \Sigma^k} \one \{b_j=\sigma\} \cdot z_{i,b} = x_{v_{i,j},\sigma} 
&\quad & \forall i \in [m], j \in [k], \sigma \in \Sigma \\
& x_{v,\sigma} \ge 0 
&\quad & \forall v \in V, \sigma \in \Sigma \\
& z_{i,b} \ge 0 
&\quad & \forall i \in [m], b \in \Sigma^k
\end{alignedat}
\]

We describe the intuition behind the linear program. First, the variable $x_{v,\sigma}$ represents a fractional assignment of the value $\sigma$ to variable $v$. Similarly, the variable $z_{i,b}$ represents a fractional assignment of the vector $b\in \Sigma^k$ to constraint $C_i$, in other words, a fractional assignment of all adjacent vertices to constraint $C_i$. The conditions on the LP are that all variables must be non-negative, the fractional assignment to any variable must add to $1$, and the fractional assignments of the constraints must agree with the fractional assignments of their neighbors. By \enquote{agree}, it means that if constraint $C_i$ is adjacent to variable $v$, then the sum of the fractional assignments to $C_i$ that assign $v$ to a given value $\sigma$ must be equal to the variable's fractional assignment for $\sigma$. As a sanity check, it is possible to prove that $\sum_{b\in \Sigma^k} z_{i,b} = 1$. We simply sum over all values of $\sigma$ for some $j\in [k]$, and this sum equals $\sum_{\sigma \in \Sigma} x_{v_{i,j},\sigma} = 1$.

\section{Reservoir Sampling}\label{sec:reservoir-sampling}

Reservoir sampling is a technique first described in \cite{vitter1985random} for sampling a set of elements uniformly without replacement from a stream. The algorithm creates a \enquote{reservoir} of size $s$, and populates the reservoir with the first $s$ elements of the stream along with a counter or the total number of elements seen. Then, for the $i^{th}$ element in the stream after the first $s$, with probability $s/i$, we replace a random index in the reservoir with our new element. This procedure maintains the property that after $n$ values are seen, any given value is in the reservoir with probability $s/n$. This is equivalent to uniformly sampling without replacement. The functions below are implementations of this process.

\begin{algorithm}
    \caption{$\InitReservoir(s)$} \label{alg:init-reservoir}
    \KwIn{Integer size $s$}
    \KwOut{Reservoir Data Structure $(S,\mathsf{count})$ containing an array of size $s$ and a counter $\mathsf{count}$}
    Initialize $S \gets $ array of size $s$. \\ $\mathsf{count} \gets 1$. \\
    \KwRet{$(S,\mathsf{count})$}
\end{algorithm}

\begin{algorithm}
    \caption{$\UpdateReservoir(R,e)$} \label{alg:update-reservoir}
    \KwIn{Reservoir $R = (S,\mathsf{count})$, and edge $e$}
    \KwOut{Updated Reservoir}
    \If{$\mathsf{count} \leq |S|$}{
        $S[\mathsf{count}] = e$
    }
    \Else{
        \WithProb{$|S|/\mathsf{count}$}{
            Select random $i$ from $1$ to $|S|$. \\
            $S[i] \gets e$ \\
        }
    }
    $\mathsf{count} \gets \mathsf{count}+1$\\
    \KwRet{$(S,\mathsf{count})$}
\end{algorithm}
\end{document}